\documentclass[10pt,twocolumn,twoside]{IEEEtran}
\pdfoutput=1
\usepackage{cite}
\usepackage{amsmath,amssymb,amsfonts,amsthm}
\newtheorem{theorem}{Theorem}[]
\newtheorem{assumption}{Assumption}[]
\newtheorem{lemma}{Lemma}[]
\newtheorem{definition}{Definition}[]
\newtheorem{corollary}{Corollary}[]
\newtheorem{remark}{Remark}[]

\usepackage{algorithm, algorithmic}

\usepackage{graphicx}
\usepackage{textcomp}
\usepackage{xcolor}

\usepackage{subfigure}
\usepackage{parskip}
\def\BibTeX{{\rm B\kern-.05em{\sc i\kern-.025em b}\kern-.08em
    T\kern-.1667em\lower.7ex\hbox{E}\kern-.125emX}}
    
\begin{document}
\title{Distributed Average Consensus via Noisy and Non-Coherent Over-the-Air Aggregation}
\author{Huiwen~Yang, Xiaomeng~Chen, Lingying~Huang, Subhrakanti~Dey, and~Ling~Shi
\thanks{Huiwen Yang, Xiaomeng Chen, and Ling Shi are with the Department of Electronic and Computer Engineering, Hong Kong University of Science and Technology, Hong Kong 00852, China (e-mail: hyangbr@connect.ust.hk;
xchendu@connect.ust.hk; eesling@ust.hk).}
\thanks{Lingying Huang is with the School of Electrical and Electronic Engineering, Nanyang Technological University, Singapore (e-mail: lingying.huang@ntu.edu.sg).}
\thanks{Subhrakanti Dey is with the Department of Electrical Engineering, Uppsala University, Uppsala, Sweden (email: subhrakanti.dey@angstrom.uu.se).}
\thanks{* Corresponding author of this work is Lingying Huang.}}
\maketitle

\begin{abstract}
    Over-the-air aggregation has attracted widespread attention for its potential advantages in task-oriented applications, such as distributed sensing, learning, and consensus. 
    In this paper, we develop a communication-efficient distributed average consensus protocol by utilizing over-the-air aggregation, which exploits the superposition property of wireless channels rather than combat it. Noisy channels and non-coherent transmission are taken into account, and only half-duplex transceivers are required. We prove that the system can \textcolor{black}{achieve average consensus in mean square and even almost surely} under the proposed protocol. Furthermore, we extend the analysis to the scenarios with time-varying topology. Numerical simulation shows the effectiveness of the proposed protocol.
\end{abstract}


\section{Introduction}
In distributed multi-agent systems, consensus problems have been studied extensively due to their wide applications~\cite{ren2007consensus, olfati2006flocking, tsianos2012consensus}. 
To achieve consensus, each agent should exchange information with their neighbors.
The research about distributed consensus span from time-invariant balanced graphs to time-varying unbalanced graphs~\cite{kingston2006discrete, li2010consensus, li2013consensus, nedic2016convergence, saldana2017resilient, zhang2020state}. To reduce communication costs, an event-triggered mechanism is introduced into the design of consensus protocols~\cite{ding2015event, ding2017overview}. 

In many practical scenarios where agents communicate with each other via wireless channels~\cite{he2013sats, kuriki2014consensus}, the studies mentioned above implicitly require the assumption that the communication links are orthogonal and the transmitted information can be decoded without error. Moreover, a set of weights should be predetermined for aggregating the received signals. 
To realize orthogonal transmission such that interference between different agents can be avoided, multiple access techniques should be adopted, such as time-division multiple access (TDMA) and frequency-division multiple access (FDMA), which require assigning different communication resource blocks to different agents. 
To achieve error-free decoding, perfect channel state information (CSI) is essential. However, the overhead of channel estimation can be very enormous, especially when there are hundreds of agents in the system. Moreover, CSI acquisition may be ravaged by pilot contamination~\cite{elijah2015comprehensive}, which defeats the protocols relying on error-free decoding. 
Although many issues in real communication systems have been considered, e.g., noises~\cite{kar2008distributed, huang2009coordination, li2010consensus}, quantization~\cite{kar2009distributed, lavaei2011quantized}, fading channels~\cite{huang2010stochastic, xu2018mean}, the interference between different agents still needs to be dealt with, which occupies numerous communication resources when the number of agents is considerable.

Recently, over-the-air aggregation, which exploits interference rather than combat it, is considered a candidate technique for task-oriented communication systems, such as the communication systems for distributed learning, sensing, and control~\cite{zhu2021over}. 
In these applications, each agent does not have to know the exact information transmitted by its neighbors, since the aggregated signals can be sufficient for the agent to extract the required information. For example, for agents performing federated learning, the required information can be the weighted summation of their neighbors' gradients~\cite{yang2020federated}. Similarly, only the convex combination of neighbors' information states is required by agents executing a consensus protocol~\cite{li2010consensus}. 
% Over-the-air aggregation also attracts much attention as one of its by-products is privacy preservation. 
Under federated settings, over-the-air aggregation has been widely utilized and studied~\cite{yang2020federated, amiri2020federated, zhu2020one, su2021data, xu2021learning}. 
In these studies, the foundation of over-the-air aggregation is coherent transmission, which requires analog amplitude modulation and channel pre-compensation~\cite{zhu2021over}. Analog amplitude modulation realizes the analog signal representation of data and channel pre-compensation eliminates the effect of heterogeneous channel fading, so that each component of a received signal corresponds to the transmitted data scaled by a pre-determined factor. Since channel pre-compensation utilizes CSI, in existing works, perfect CSI is still needed for coherent transmission. 
\textcolor{black}{
Compared with traditional communication with multiple access techniques that aim to avoid interference among agents, over-the-air aggregation has the following advantages:
\begin{itemize}
    \item To enhance network capacity: By allowing all agents to transmit in the same communication resource block, it can be expected that less latency will be caused and more agents will be accommodated by the communication system. Furthermore, using analog transmission, coding-decoding delays are avoided, with only electromagnetic wave propagation delay to account for.
    \item To improve reliability: When multiple signals are aggregated coherently, the received signal power can be combined, while the noise power remains relatively constant. As a result, the overall signal-to-noise ratio (SNR) can be improved. A higher SNR generally leads to better signal quality and a more robust communication link, as the received signal is less susceptible to interference or noise. 
    \item To protect privacy: The signal distortion caused by fading and noisy channels is a mask for free that can protect data privacy. This can help in obfuscating individual data transmissions, making it more challenging for unauthorized entities to isolate and identify specific agent data. 
\end{itemize}
}
% By over-the-air aggregation, the communication resources can be saved since all agents are permitted to transmit over the same communication resource block. 
\textcolor{black}{Although over-the-air aggregation under federated settings, where systems are assisted by a central node, has been extensively studied, there are limited works investigating fully decentralized systems with over-the-air aggregation.} In~\cite{molinari2018exploiting}, Molinari~et~al. propose a consensus protocol, under which systems can achieve weighted average consensus via over-the-air aggregation. In~\cite{molinari2021max}, Molinari~et~al. develop a max-consensus protocol based on over-the-air aggregation. However, the main results of both~\cite{molinari2018exploiting} and~\cite{molinari2021max} are based on some unrealistic assumptions. First, each agent is assumed to be equipped with a full-duplex transceiver, which, however, has been relaxed in their latest work~\cite{molinari2022over} by clustering agents. Second, it is assumed that the channels are noiseless and the receivers are noise-free. Third, the transmitters are assumed to be coarsely synchronized.
In real communication systems, noises are difficult to eliminate, and most transceivers only operate in half-duplex mode. What is more, transmitter synchronization, and hence coherent transmission, is much more difficult to realize under a fully decentralized setting than under a federated setting. The reason is that there is only one destination node, i.e., the central node, in a system under the federated setting. Therefore, each edge node only needs to perform phase cancellation once at each time step. 
However, phase cancellation is a tricky task in fully decentralized systems. Since each node communicates with its neighbors over different channels, which cause different phase shifts, it should transmit different signals with different phase cancellation terms to its different neighbors, respectively. Obviously, over-the-air aggregation will be overshadowed by the constraint of coherent transmission. As a result, to take full advantage of over-the-air aggregation, it is necessary to investigate non-coherent transmission. 
Recently, Michelusi~\cite{michelusi2022decentralized}  investigates the system design for decentralized federated learning with non-coherent over-the-air aggregation, where the communication topology is modeled by a series of complete bipartite graphs, and the mean square convergence of the learning algorithm is proved.

In this paper, we consider the average consensus problem in multi-agent systems via over-the-air aggregation. We jointly design the communication mechanism and the consensus protocol which takes noises and asynchronous transmitters into account, and present simulation results to show the effectiveness of the proposed scheme. 
The main contributions are summarized as follows.
\begin{itemize}
    \item[1)] We propose a communication-efficient distributed average consensus protocol by utilizing over-the-air aggregation. Compared with~\cite{molinari2018exploiting, molinari2021max, molinari2022over}, the implementation of the proposed protocol does not require the assumption of noiseless channels, full-duplex transceivers and phase synchronization of transmitters. Moreover, the effect of noise is taken into account. In summary, our protocol is more practical than previous works.
    \item[2)] We analyze the convergence of the proposed protocol, and prove that the multi-agent system can achieve mean square average consensus and almost sure consensus with a suitable choice of decreasing stepsizes. 
    \textcolor{black}{The convergence analysis is tricky due to two main reasons. First, non-coherent transmission introduces more noise that is state-dependent and has a more complex form. Second, the resulting weight matrix is doubly stochastic in expectation, while existing works at least require the weight matrix to be row-stochastic almost surely and column-stochastic in expectation.}
    % Compared with~\cite{}, which only considered state-independent noise, \cite{}, where the noises depend on the consensus error, and~\cite{}, where the noise imposed to each agent depends on the state of itself, the interference terms induced by non-coherent transmission are dependent to the products of different agents' information states. 
    \item[3)] We further investigate the convergence performance of the system when the communication topology is time-varying. Specifically, we prove that the system can achieve mean square average consensus and almost sure convergence when the time-varying graph is jointly connected.
    \item[4)] We present numerical simulation results to validate the effectiveness of the proposed protocol. The simulation results show that the system is robust to noises under the proposed protocol. 
    
\end{itemize}

The remainder of this paper is organized as follows.
Section \uppercase\expandafter{\romannumeral2} describes the system design. 
Section \uppercase\expandafter{\romannumeral3} analyses the convergence performance of the system.
Section \uppercase\expandafter{\romannumeral4} extends the main results to the scenario with time-varying topology. 
Section \uppercase\expandafter{\romannumeral5} presents some numerical simulations. 
Section \uppercase\expandafter{\romannumeral6} concludes this paper.
	
\emph{Notations:} 
Let $\mathbb{R}$, $\mathbb{R}^n$, and $\mathbb{R}^{n\times m}$ denote the set of real numbers, the $n$-dimensional Euclidean space, and the set of real matrices with size $n\times m$, respectively. Denote the set of complex numbers as $\mathbb{C}$.
For a real number $r$, $|r|$ denotes its absolute value.
For a complex $c$, $|c|$ denotes its norm (amplitude) and $\mathrm{Re}[c]$ and $\mathrm{Im}[c]$ denotes its real part and imaginary part, respectively. 
For a vector $\boldsymbol{v}\in\mathbb{R}^n$ and a matrix $M\in\mathbb{R}^{n\times m}$, their transposes are denoted by $\boldsymbol{v}^T$ and $M^T$, respectively. \textcolor{black}{For the matrix $M$, $\Vert M \Vert$ denotes its 2-norm and $\Vert M \Vert_F$ denotes its Frobenius norm.}  
Moreover, $v_i$ denotes the $i$-the entry of vector $\boldsymbol{v}$, and $\Vert \boldsymbol{v} \Vert$ denotes the standard Euclidean norm of $\boldsymbol{v}$.
The $n$-dimensional column vector with all elements being $1$ is denoted by $\boldsymbol{1}_n$. 
Let $\{X(k)\}_{k\geq 0}$ denote the sequence $X(0), X(1), \ldots$, where $X(k),\forall k\geq 0$ can be matrices or vectors. Let $\lfloor \cdot \rfloor$ denote the floor function.

\section{System Description}\label{sec:system_description}
\subsection{Multi-agent Systems}
We consider a multi-agent system with $N$ agents and model the network topology as a connected time-invariant undirected graph $\mathcal{G} \triangleq (\mathcal{V}, \mathcal{E})$,
where $\mathcal{V}=\{1, 2, \ldots, N\}$ is the agent set and $\mathcal{E}\in\{(i, j)| i, j\in\mathcal{V}\}$ is the edge set. If agent $i$ and agent $j$ can communicate with each other, we have $(i,j)\in\mathcal{E}$ and $(j,i)\in\mathcal{E}$. 
\textcolor{black}{It is assumed that agents will not communicate with themselves, i.e., $(i,i)\notin\mathcal{E},\forall i\in\mathcal{V}$. 
Define the potential neighbor set of agent $i$ as $\widetilde{\mathcal{N}}_i\triangleq\{j\in\mathcal{V}|(i,j)\in\mathcal{E} \text{ and } (j, i)\in\mathcal{E} \}$.}
We will investigate time-varying network topology and relax the connectivity assumption on the graph in section~\ref{sec:extension}.

We consider that each agent $i\in\mathcal{V}$ has a bounded initial information state 
% $x_i(0)\in\mathbb{R}$.
$x_i(0)\in[x_\mathrm{min},x_\mathrm{max}]\subset\mathbb{R}$. 
\textcolor{black}{In most cases, the bounds of the initial information states can be known based on prior knowledge. For example, the velocity of the agents should be within a finite range. }
Without loss of generality, we consider scalar cases here, but the main results can be easily extended to vector cases \textcolor{black}{(see Remark~\ref{rmk:vector})}. 
The aim of the system is to achieve average consensus by letting agents iteratively exchange information with each other and update their information state according to a pre-designed consensus protocol. 
To characterize the asymptotic behavior of the agents, we introduce the following definitions.
\textcolor{black}{
\begin{definition}[Weak consensus~\cite{huang2009coordination}]
The multi-agent system is said to achieve weak consensus if 
\begin{equation}
    \lim_{k\rightarrow\infty} \mathbb{E}\left[\left(x_i(k)-\frac{\mathbf{1}_N^T}{N}\boldsymbol{x}(k)\right)^2\right]=0, \forall i\in\mathcal{V}. 
\end{equation}
\end{definition}
}
\begin{definition}[Mean square average consensus~\cite{li2010consensus}]
The multi-agent system is said to achieve mean square average consensus if there exists a random variable $x^*$ such that
\begin{equation}
    \lim_{k\rightarrow\infty} \mathbb{E}[(x_i(k)-x^*)^2]=0, \forall i\in\mathcal{V},
\end{equation}
where $x^*$ satisfies $\mathbb{E}[x^*]=\frac{1}{N}\sum_{i\in\mathcal{V}}x_i(0)$ and $Var(x^*)<\infty$.
\end{definition}
\begin{definition}[Almost sure consensus~\cite{li2010consensus}]
The multi-agent system is said to achieve almost sure consensus if there exists a random variable $x^*$ such that 
\begin{equation}
    \lim_{k\rightarrow\infty} x_i(k) = x^* a.s., \forall i\in\mathcal{V}.
\end{equation}
% where $x^*$ satisfies $\mathbb{E}[x^*]=\frac{1}{N}\sum_{i\in\mathcal{V}}x_i(0)$ and $Var(x^*)<\infty$.
\end{definition}
\textcolor{black}{
\begin{remark}
    Mean square average consensus ensures that the information state $x_i(k)$ is the asymptotically unbiased estimate of $\frac{1}{N}\sum_{i\in\mathcal{V}}x_i(0)$. 
    Almost sure consensus ensures that the information state of each agent can converge to the same random variable with probability $1$. 
    If both of them are achieved, $x_i(k), \forall i\in\mathcal{V}$ will converge to the same random variable $x^*$ with $\mathbb{E}[x^*]=\frac{1}{N}\sum_{i\in\mathcal{V}}x_i(0)$ and $Var(x^*)<\infty$ with probability $1$.
\end{remark}
}

\subsection{Communication Mechanism}
% Due to the broadcast nature of the wireless medium, each agent has the capability to receive the signals transmitted by all other agents. 
In~\cite{molinari2018exploiting} and~\cite{molinari2021max}, each agent is assumed to have a full-duplex transceiver such that it can transmit and receive signals simultaneously. However, this assumption will introduce significantly higher complexity and cost, and hence should be relaxed for the implementation in real systems. Therefore, in this paper we assume each agent only has a half-duplex transceiver. Moreover, each time step is divided into two time slots, i.e., $s_1$ and $s_2$, and each agent will randomly select one time slot to broadcast its messages and use the other one to receive aggregated signals. Define $\gamma_i(k)$ as follows
\begin{equation}\label{eq:gamma}
    \gamma_i(k)=\left\{
    \begin{aligned}
        &1,\quad\text{if agent $i$ selects $s_1$ for transmission,}\\ 
        &0,\quad\text{\color{black}if agent $i$ selects $s_2$ for transmission.}
    \end{aligned}
    \right.
\end{equation}
At each time step, agent $i$ will select time slot $s_1$ to transmit with probability $p_i\in(0,1)$, i.e., $\mathbb{P}(\gamma_i(k)=1)=p_i$, and hence $\mathbb{P}(\gamma_i(k)=0)=1-p_i$. Moreover, the random variables $\gamma_i(k),\forall i\in\mathcal{V}, k\geq 0$ are independent with each other (i.i.d).
\textcolor{black}{As a result, the actual communication topology is time-varying and not necessarily connected at every time step even the physical network topology is time-invariant and connected (see Fig.~\ref{fig:topology}).}

\begin{figure}[h]
    \centering
    \includegraphics[width=0.9\linewidth]{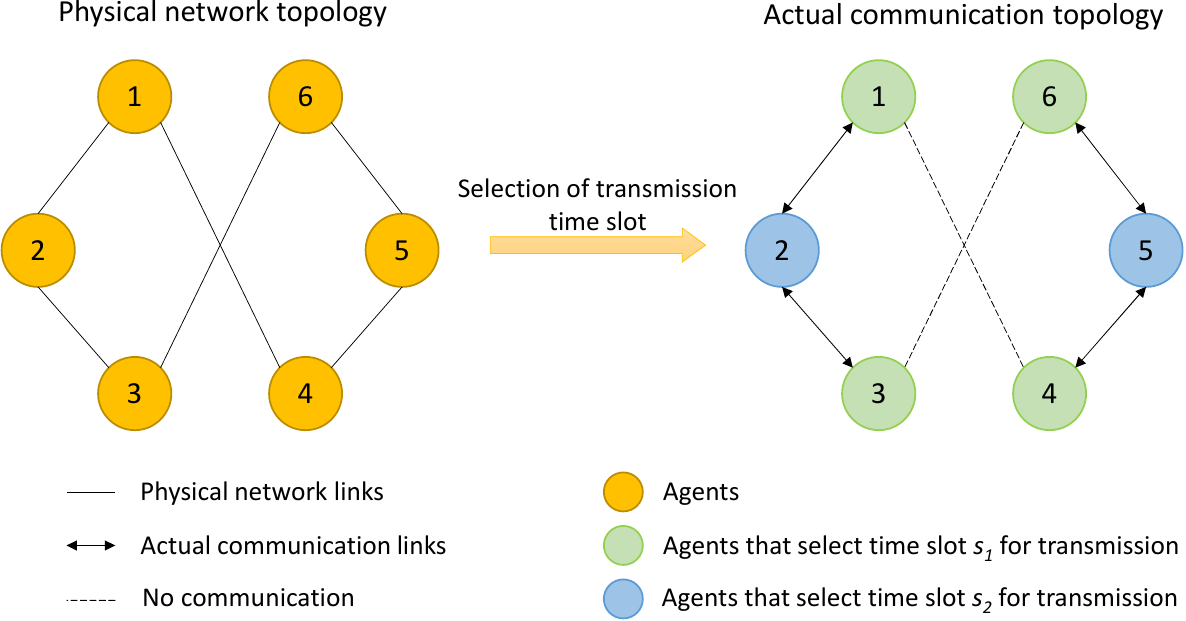}
    \caption{\textcolor{black}{Example of a physical network topology and its actual communication topology.}}
    \label{fig:topology}
\end{figure}

% Under this communication mechanism, the communication topology can is a time-varying directed graph $\mathcal{G}(k)\triangleq \left(\mathcal{V},\mathcal{E}(k)\right)$, where $\mathcal{V}=\{1,2,\ldots,N\}$ denotes the node (agent) set and $\mathcal{E}(k)\subset\mathcal{V}\times\mathcal{V}$ denotes the time-varying edge set. If $(i,j)\in\mathcal{E}(k)$, there exists a communication link from agent $i$ to agent $j$ at time $k$, i.e., agent $i$ can transmit signals to agent $j$ at time $k$. 
% Moreover, due to the half-duplex assumption, agents will not transmit signals to themselves, i.e., $(i,i)\notin\mathcal{E}(k), \forall i\in\mathcal{V}, k\geq 0$. 
% Define the in-neighbor set of node $i$ as $\mathcal{N}_i^\mathrm{in}(k)\triangleq\{j\in\mathcal{V}|(j,i)\in\mathcal{E}(k)\}$, and the out-neighbor set as $\mathcal{N}_i^\mathrm{out}(k)\triangleq\{j\in\mathcal{V}|(i.j)\in\mathcal{E}(k)\}$. According to~\eqref{eq:gamma}, we have
% \begin{align}
%     \mathcal{N}_i^\mathrm{in}(k)&= \{ j\in\mathcal{V}| \gamma_j(k)=1-\gamma_i(k)\}\\
%     \mathcal{N}_i^\mathrm{out}(k)&= \{ j\in\mathcal{V}| \gamma_j(k)=1-\gamma_i(k)\}
% \end{align}

\begin{figure}[h]
	\centering
	\includegraphics[width=0.9\linewidth]{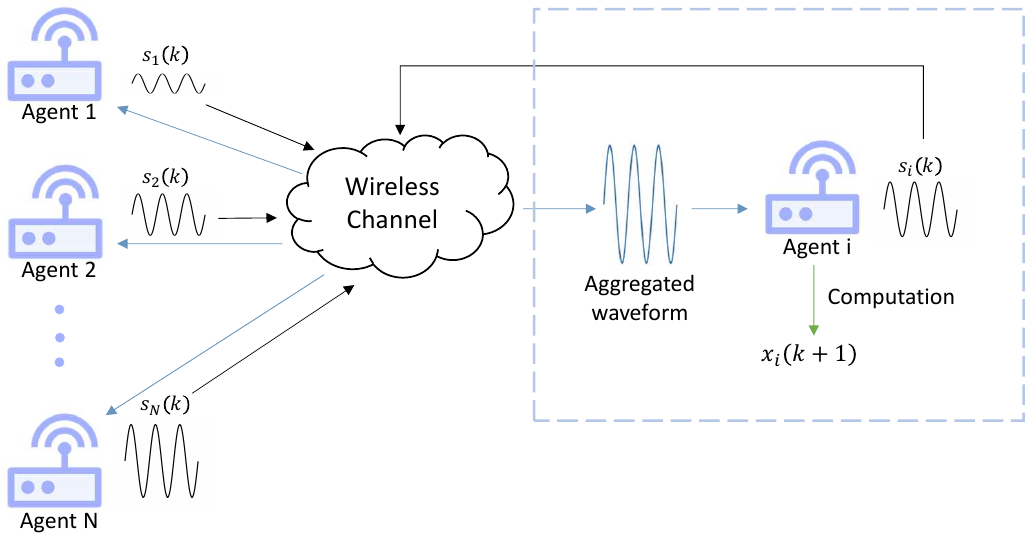}
	\caption{\textcolor{black}{Over-the-air aggregation.}}
	\label{fig:aggregation}
\end{figure}

\subsection{Noisy and Non-coherent Over-the-air Aggregation}

\textcolor{black}{Over-the-air aggregation is an information aggregation approach, which exploits the superposition property of wireless channels. It allows all agents to transmit analog signals carrying the information of their states in the same communication resource block. The information states of each agent are simulated by the variation (for example, in amplitude) of a sine wave. Then, all the signals will form an aggregated waveform, which involves the information carried by these signals. The} \textcolor{black}{aggregated waveform is received by the receiver of a target agent and used for computation (see Fig.~\ref{fig:aggregation}).}

\textcolor{black}{
In the literature exploring different consensus protocols, the analyses are all based on the setting of traditional multiple access schemes, e.g., TDMA and FDMA. The core idea of these traditional multiple-access schemes, which differs from over-the-air aggregation, is to avoid transmission interference between multiple agents by allocating distinct communication resource blocks, e.g., time slot and frequency bandwidth, to each agent.
That is to say, the information transmitted by all agents will be separately decoded by their neighbors. When the number of agents is considerable, the limited communication } \textcolor{black}{resources will become a bottleneck. 
Compared with traditional multiple access schemes, over-the-air aggregation can save more communication resources (see Fig.~\ref{fig:cmp}), and hence can enhance network capacity.}
\begin{figure}
    \centering
    \includegraphics[width=0.9\linewidth]{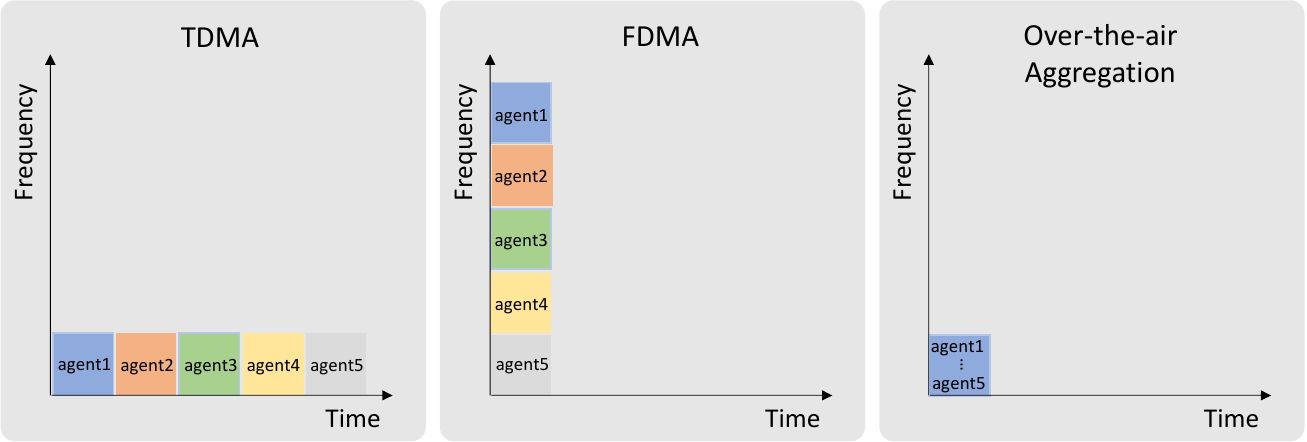}
    \caption{\textcolor{black}{Resource allocations of traditional multiple access techniques and Over-the-air aggregation.}}
    \label{fig:cmp}
\end{figure}

\textcolor{black}{Among almost all the previous works investigating over-the-air aggregation, transmitters are assumed to be perfectly (phase) synchronized, i.e., transmitters perform coherent transmission. 
To realize phase synchronization and coherent transmission, continuous channel estimation is required. Moreover, if one transmitter has multiple target receivers, it needs to do phase synchronization for each target receiver since the channels between the transmitter and its distinct targets can be different, and hence can induce different phase shifts. 
This is feasible for federated systems because all the transmitters only have one target receiver, i.e., the central server. However, this is difficult for fully decentralized systems since each agent (transmitter) usually has more than one neighbor (receiver). In such scenarios, over-the-air aggregation will lose its superiority in communication efficiency. As a result, we chose to design the system under a \textbf{non-coherent transmission} setting, where frequent phase synchronization is avoided, but waveform distortion is introduced (see Fig.~\ref{fig:coherent} and Fig.~\ref{fig:noncoherent}). }
\begin{figure}[h]
\centering
\begin{minipage}[t]{0.4\linewidth}
    \centering
    \includegraphics[width=0.9\linewidth]{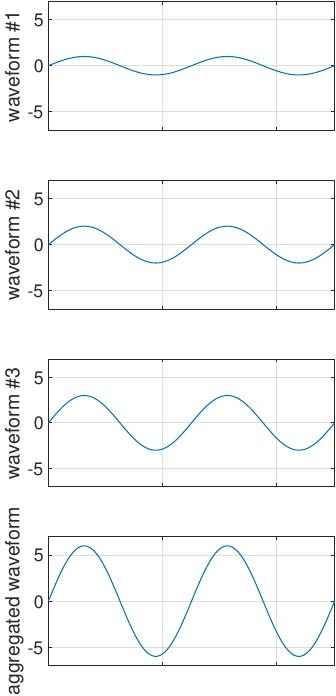}
    \caption{Waveforms of coherent transmission (waveform \#1 - \#3 have the same phase.)}
    \label{fig:coherent}
\end{minipage}
\qquad
\begin{minipage}[t]{0.4\linewidth}
    \centering
    \includegraphics[width=0.9\linewidth]{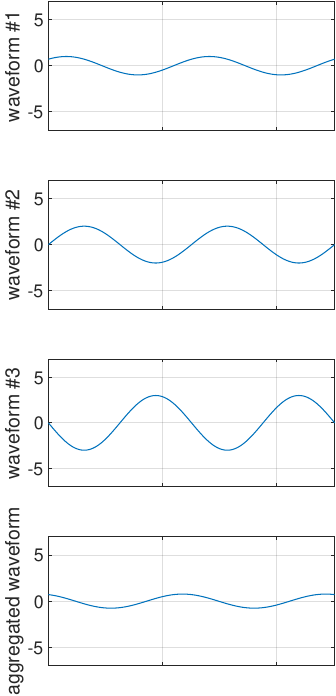}
    \caption{Waveforms of non-coherent transmission (waveform \#1 - \#3 have different phases.)}
    \label{fig:noncoherent}
\end{minipage}
\end{figure}

In the following, we consider that each agent is equipped with one transmit antenna and one receive antenna. Let $h_{ij}(k)\sim\mathcal{CN}(0,\Lambda_{ij})$ denote the fading channel between the receive antenna of agent $i$ and the transmit antenna of agent $j$ at time $k$. We assume that $h_{ij}(k), \forall i, j\in\mathcal{V}, k\geq 0$ are independent and $\Lambda_{ij}=\Lambda_{ji}$.

The signal transmitted by agent $i$ at time $k$ is 
\begin{equation}\label{eq:signal}
    s_i(k) = \sqrt{\rho x_i(k)}u,
\end{equation}
where $u\in\{c\in\mathbb{C}\ |\ |c|=1\}$ characterizes the carrier wave, and $\rho$ is the transmission coefficient for transmission power control.
\textcolor{black}{
\begin{remark}
    The transmitted signal in~\eqref{eq:signal} is an analog signal, which is different from a digital signal. The information state $x_i(k)$ is directly represented by the amplitude of the analog signal via signal power control. In contrast, digital transmission requires sampling, quantization, and coding~\cite{haykin2008communication}. 
    % Note that since $x_i(k)$ can be negative, $\sqrt{x_i(k)}$ is a complex number, whose amplitude and phase are $\left|\sqrt{x_i(k)}\right|$ and $\angle_i(k)$, respectively. 
\end{remark}
}
\begin{remark}
    The amplitude and phase of $u$ characterize the amplitude and phase of the carrier wave.
    Here, we let $|u|=1$ to facilitate power control. 
    % \textcolor{black}{If $x_i(k)$ is negative, we have $\angle_i(k) = \frac{\pi}{2}$ and hence the the carrier wave becomes $u\angle_i(k)$, i.e., $u$ with $\frac{\pi}{2}$ phase shift.}  
    \textcolor{black}{Moreover, we assume $x_\mathrm{min} = 0$ in the subsequent analysis. Note that when $x_\mathrm{min}\ne 0$, one can easily subtract $x_\mathrm{min}$ from $x_i(0)$. Then, for the cases with $x_\mathrm{min}\ne 0$, the analysis remains unchanged except that $x_\mathrm{min}$ needs to be added back. }
    % \textcolor{black}{Moreover, we assume 
    % $x_\mathrm{min}=0$ 
    % % $x_\mathrm{min}\triangleq \min_{i\in\mathcal{V}}\inf_{k\geq 0}x_i(k) \geq 0$ 
    % in the subsequent analysis. Note that when $x_\mathrm{min} \ne 0$, one can easily subtract $x_\mathrm{min}$ from $x_i(0)$. Then, for the cases with $x_\mathrm{min} \ne 0$, the analysis remains unchanged except that $x_\mathrm{min}$ needs to be added back.}
    \textcolor{black}{For the power control of the carrier wave, $x_i(k), \forall i\in\mathcal{V}, k\geq 0$ should be non-negative. However, it is possible that some of the information states is negative at some time. The probability that an information state turns to be negative can be reduced by adopting smaller stepsize or adding a larger offset to the initial information states so that $x_\textrm{min}$ is much larger than 0.}
\end{remark}
\textcolor{black}{
\begin{remark}\label{rmk:vector}
    If $x_i(k)$ is a $d$-dimensional vector, we can consider that all the agents need to achieve consensus on $d$ different values at each time step. As a result, all the subsequent analysis remains applicable as long as the different entries of $x_i(k)$ can be transmitted without interfering with each other, e.g., the entries of $x_i(k)$ are sequentially transmitted.
\end{remark}
}

\textcolor{black}{Define 
% the potential neighbor set of agent $i$ as $\mathcal{V}\triangleq\{j\in\mathcal{V} | (j, i) \in\mathcal{E}\}$and 
the actual neighbor set of agent $i$ at time $k$ as $\mathcal{N}_i(k)\triangleq\{j\in\tilde{\mathcal{N}}_i | \gamma_j(k) = 1-\gamma_i(k)\}$.}
Then, the signal received by agent $i$ can be expressed as
\begin{equation}
\begin{split}
    y_i(k) &= \sum_{j\in\mathcal{N}_i(k)} h_{ij}(k) s_{j}(k) + n_i(k)\\
    &=\sum_{j\in\widetilde{\mathcal{N}}_i}\Gamma_{ij}(k) h_{ij}(k) s_{j}(k) + n_i(k),
\end{split}
\end{equation}
where $\Gamma_{ij}(k)\triangleq \gamma_i(k)(1-\gamma_j(k))+\gamma_j(k)(1-\gamma_i(k))$ and $n_i(k)\sim\mathcal{CN}(0,\sigma_i^2)$ represents the additive white Gaussian noise (AWGN). Note that $\Gamma_{ii}(k)\equiv 0$. 
\textcolor{black}{
\begin{remark}
    We consider that the channel noise $n_i(k)$ is a complex Gaussian random variable since it is a common assumption in the wireless community. 
    In our proof, besides the independence of $n_i(k), \forall i\in\mathcal{V}, k\geq 0$ and $h_{ij}(k), \forall i,j\in\mathcal{V}, k\geq 0$, we only utilize the property that the real and imaginary parts of $n_i(k)$ are independent random variables with zero mean and bounded variance. Therefore, the main results hold as long as the noise has the mentioned properties.
\end{remark}
}

After receiving, agent $i$ can obtain the amplitude of the received signal
\begin{equation}\label{eq:received_signal}
\begin{split}
    &\left|y_i(k)\right|^2 
    \\&= \left|\sum_{j\in\widetilde{\mathcal{N}}_i}\Gamma_{ij}(k) h_{ij}(k) s_{j}(k) + n_i(k)\right|^2
    \\&=\left(\sum_{j\in\widetilde{\mathcal{N}}_i}\sqrt{\rho x_j(k)}\Gamma_{ij}(k)\mathrm{Re}\left[h_{ij}(k)u\right]+\mathrm{Re}[n_i(k)]\right)^2\\&\quad +\left(\sum_{j\in\widetilde{\mathcal{N}}_i}\sqrt{\rho x_j(k)}\Gamma_{ij}(k)\mathrm{Im}\left[h_{ij}(k)u\right]+\mathrm{Im}[n_i(k)]\right)^2
    \\&=\sum_{j\in\widetilde{\mathcal{N}}_i}\rho x_j(k)\Gamma_{ij}(k)|h_{ij}(k)|^2 + |n_i(k)|^2 \\&\quad + \sum_{j\in\widetilde{\mathcal{N}}_i}\sum_{l\in\widetilde{\mathcal{N}}_i\backslash j}Y_{i,jl}^{(1)}(k) + 2\sum_{j\in\widetilde{\mathcal{N}}_i}Y_{ij}^{(2)}(k),
\end{split}
\end{equation}
where
\begin{align*}
    Y_{i,jl}^{(1)}(k) \triangleq& \rho \sqrt{x_j(k) x_l(k)} \Gamma_{ij}(k)\Gamma_{il}(k) \left(\mathrm{Re}\left[h_{ij}(k)u\right]\right. \\&\left.\times\mathrm{Re}\left[h_{il}(k)u\right] +\mathrm{Im}\left[h_{ij}(k)u\right] \mathrm{Im}\left[h_{il}(k)u\right]\right),\\
    Y_{ij}^{(2)}(k) \triangleq& \sqrt{\rho x_j(k)}\Gamma_{ij}(k) \left(\mathrm{Re}\left[h_{ij}(k)u\right] \mathrm{Re}[n_i(k)] \right.\\&\left. + \mathrm{Im}\left[h_{ij}(k) u\right] \mathrm{Im}[n_i(k)]\right).
\end{align*}
Note that the third step is based on two facts:
1) $\Gamma_{ij}^2(k) = \Gamma_{ij}(k)$, and 2) for two complex number $c_1$ and $c_2$, $|c_1 c_2| = |c_1| |c_2|$.
\begin{remark}
    The second term and the fourth term of~\eqref{eq:received_signal} are introduced by the AWGN, and the third term is introduced by the non-coherent transmission. Note that $Y_{i,jl}^{(1)}(k)$ and $Y_{ij}^{(2)}(k)$ are dependent on the information states, which is more difficult to handle than state-independent noises~\cite{kar2008distributed}. 
\end{remark}

Define 
\begin{align*}
    \boldsymbol{x}(k) &\triangleq [x_1(k), x_2(k), \ldots, x_N(k)]^T,\\
    \boldsymbol{\gamma}(k) &\triangleq [\gamma_1(k), \gamma_2(k), \ldots, \gamma_N(k)]^T,\\
    \boldsymbol{n}(k) &\triangleq [n_1(k), n_2(k), \ldots, n_N(k)]^T,\\
    H(k) &\triangleq [h_{ij}(k)],
    % \begin{bmatrix}
    % h_{11}(k) & h_{12}(k) & \cdots & h_{1N}(k) \\
    % h_{21}(k) & h_{22}(k) & \cdots & h_{2N}(k) \\
    % \vdots & \vdots & \cdots & \vdots \\
    % h_{N1}(k) & h_{N2}(k) & \cdots & h_{NN}(k) \\
    % \end{bmatrix}
\end{align*}
and define an increasing sequence of $\sigma$-fields $\mathcal{F}_k$ as
\begin{equation}
\begin{split}
    \mathcal{F}_k\triangleq &\left\{\boldsymbol{x}(0), \ldots, \boldsymbol{x}(k),  \boldsymbol{\gamma}(0), \ldots, \boldsymbol{\gamma}(k-1), H(0), \ldots, \right. \\& \left. H(k-1), \boldsymbol{n}(0), \ldots, \boldsymbol{n}(k-1)\right\},
\end{split}
\end{equation}
then we have
\begin{align}
    \mathbb{E}[Y_{i,jl}^{(1)}(k)|\mathcal{F}_k] = 0,\label{eq:e_Y1}\\
    \mathbb{E}[Y_{ij}^{(2)}(k)|\mathcal{F}_k] = 0,\label{eq:e_Y2}
\end{align}
and hence 
\begin{equation}
\begin{split}
    \mathbb{E}\left[\left|y_i(k)\right|^2| \mathcal{F}_k\right] &= \mathbb{E}\left[\left.\sum_{j\in\mathcal{V}} a_{ij}(k)x_j(k) + |n_i(k)|^2 \right|\mathcal{F}_k \right]\\&=\sum_{j\in\mathcal{V}} \bar{a}_{ij} x_j(k) + \sigma_i^2,
\end{split}
\end{equation}
\textcolor{black}{where 
    $$a_{ij}(k) =\left\{
    \begin{aligned}
        &\rho \Gamma_{ij}(k)|h_{ij}(k)|^2, \forall j \in\widetilde{\mathcal{N}}_i,\\ 
        & 0, \forall j \notin\widetilde{\mathcal{N}}_i,
    \end{aligned}
    \right.$$ 
and 
    $$\bar{a}_{ij} =\left\{
    \begin{aligned}
        &\rho\left(p_i(1-p_j)+p_j(1-p_i)\right)\Lambda_{ij}, \forall j\in\widetilde{\mathcal{N}}_i,\\ 
        & 0, \forall j \notin\widetilde{\mathcal{N}}_i.
    \end{aligned}
    \right.$$
}
Moreover, we have $a_{ii}(k)=0$. \textcolor{black}{Sensor $i$ is assumed to have the knowledge of the statistics of the channel noise $n_i(k)$ and the fading $h_{ji}(k), \forall j\in\textcolor{black}{\widetilde{\mathcal{N}}_i}$, i.e., $\sigma_i^2$ and $\Lambda_{ji}, \forall j\in\textcolor{black}{\widetilde{\mathcal{N}}_i}$ are known by sensor $i$. The probabilities $p_i, \forall i\in\textcolor{black}{\widetilde{\mathcal{N}}_i}$ are pre-determined and known by all the sensors. Since it has been assumed that $\Lambda_{ij} = \Lambda_{ji}$, sensor $i$ has full knowledge of $\bar{a}_{ij}, \forall j\in\textcolor{black}{\widetilde{\mathcal{N}}_i}$.}

\textcolor{black}{
    \begin{remark}
        Generally, achieving average consensus over a directed graph requires that each agent at least knows its neighbors’ out-degree due to the requirement of column-stochastic matrices. Under the over-the-air aggregation setting, the out-degree of an agent consists of the channel statistics from it to all its out-neighbors, which is difficult to obtain by its in-neighbors while considering a directed graph. Therefore, it remains a challenging problem to consider directed graphs.
    \end{remark}
}
% \textcolor{black}{
% \begin{remark}
%     Here we consider that the network topology graph is a complete graph, i.e., all agents communicate in a broadcast manner, and each may be a potential neighbor to the other. As a result, each agent requires global information, i.e., $\Lambda_{ji}, \forall j\in\mathcal{V}$ and $p_i, \forall i\in\mathcal{V}$.   
% \end{remark}
% }
% In the following, we consider two scenarios.
% \subsubsection*{Scenario 1}
% $h_{ij}(k)$ are i.i.d., i.e., $h_{ij}(k)\sim\mathcal{CN}(0,\Lambda)$, and $\gamma_i(k)$ are also i.i.d., i.e. $p_i=p, \forall i\in\mathcal{V}$. 
% In this case, we design 
% \begin{equation}
%     \rho=\frac{1}{4\Lambda p(1-p)(N-1)}.
% \end{equation}

% We design
% \begin{equation}
%     \rho = \frac{1}{\max_{l\in\mathcal{V}} \sum_{j\in\mathcal{V}} \left(p_l(1-p_j)+p_j(1-p_l)\right)\Lambda_{lj}}.
% \end{equation}

% \begin{remark}

% \end{remark}

\subsection{Consensus Protocol}
\textcolor{black}{In this subsection, we propose a consensus protocol under which the system can achieve mean square average consensus and almost sure consensus (see Section~\ref{sec:convergence}). }

The agents update their information state according to the following consensus protocol.
\begin{equation}\label{eq:protocol}
    x_i(k+1) = \left(1 - \alpha(k) \sum_{j\in\mathcal{V}}\bar{a}_{ij} \right) x_i(k) + \alpha(k)\left(|y_i(k)|^2-\sigma_i^2\right),
\end{equation}
where $\alpha(k)>0$ is the time-varying step size. 

% \begin{assumption}
%     The network topology graph $\mathcal{G}$ is connected.
% \end{assumption}

\begin{assumption}\label{asm:stepsize}
    The step size $\alpha(k)$ in~\eqref{eq:protocol} satisfies the following conditions:
    \begin{itemize}
        \item[a)] $\sum_{k=0}^\infty \alpha(k)=\infty,\quad \sum_{k=0}^\infty \alpha^2(k)<\infty.$
        \item[b)] $1-\alpha(k) \sum_{j\in\mathcal{V}}\bar{a}_{ij}>0, \forall i\in\mathcal{V}, k\geq 0.$
        % \item[c)] 
        % \textcolor{black}{$\max_{i,j\in\mathcal{V}}\left|\alpha(k)-\alpha(k)\right| = o\left(\max_{i\in\mathcal{V}}\alpha(k)\right), \forall k\geq 0$.}
    \end{itemize}
\end{assumption}
\begin{remark}
    Condition a) is commonly used in distributed optimization~\cite{li2018distributed,xie2018distributed}
    and distributed consensus~\cite{huang2009coordination,huang2009stochastic,huang2010stochastic}. Moreover, condition b) is easy to be satisfied. 
    For example, we can set $\alpha(k) = \frac{1}{\max_{i\in\mathcal{V}}\sum_{j\in\mathcal{V}}\bar{a}_{ij} k^p}, \forall k\geq 0 (p\in(0.5,1])$. 
    Note that the choice of the step size $\alpha(k)$ will influence the convergence rate.
\end{remark}
% \begin{remark}
%     It is reasonable to design $\rho$ without introducing time-varying values or unknown information which needs to be acquired. One possible design is $\rho = \frac{1}{\max_{l\in\mathcal{V}} \sum_{j\in\mathcal{V}} \left(p_l(1-p_j)+p_j(1-p_l)\right)\Lambda_{lj}}$, which only require the information which is also needed by~\eqref{eq:protocol}. 
% \end{remark}

Define
\begin{align*}
    v_i(k) &\triangleq \left(|n_i(k)|^2-\sigma_i^2\right) + \sum_{j\in\mathcal{V}}\sum_{l\in\mathcal{V}\backslash j}Y_{i,jl}^{(1)}(k) + 2\sum_{j\in\mathcal{V}}Y_{ij}^{(2)}(k),\\
    \boldsymbol{v}(k) &\triangleq \left[ v_1(k), v_2(k), \ldots, v_N(k) \right]^T,\\
    \mathcal{A}(k) &\triangleq [a_{ij}(k)],\\
    \mathcal{D}(k) &\triangleq  \mathrm{diag}\left(\sum_{j\in\mathcal{V}}\bar{a}_{1j}, \sum_{j\in\mathcal{V}}\bar{a}_{2j}, \ldots, \sum_{j\in\mathcal{V}}\bar{a}_{Nj}\right),\\
    \mathcal{L}(k) &\triangleq \mathcal{D}(k)-\mathcal{A}(k),
    % \\ \textcolor{black}{A(k)} &  \textcolor{black}{\triangleq  \mathrm{diag}\left(\alpha_1(k), \alpha_2(k), \ldots, \alpha_N(k)\right)},
\end{align*}
then the protocol~\eqref{eq:protocol} can be written as a compact matrix form as follows
\begin{equation}\label{eq:protocol_compact}
   \boldsymbol{x}(k+1) = \left(I_N - \alpha(k)\mathcal{L}(k) \right) \boldsymbol{x}(k) + \alpha(k)\boldsymbol{v}(k).
\end{equation}
Define $\bar{\mathcal{L}}(k) \triangleq \mathbb{E}[\mathcal{L}(k)]$ and $\Delta\mathcal{L}(k) \triangleq \bar{\mathcal{L}}(k) - \mathcal{L}(k)$. Since $\gamma_i(k), \forall i\in\mathcal{V}$ and $h_{ij}(k), \forall i,j\in\mathcal{V}$ are i.i.d. random processes, $\bar{\mathcal{L}}(k), \forall k\geq 0$ are the same constant, which we denote as $\bar{\mathcal{L}}$ in the subsequent analysis. 
Then \eqref{eq:protocol_compact} can be rewritten as
\begin{equation}
   \boldsymbol{x}(k+1) = \left(I_N - \alpha(k)\bar{\mathcal{L}} \right) \boldsymbol{x}(k) + \alpha(k)\boldsymbol{w}(k).
\end{equation}
    where $\boldsymbol{w}(k)\triangleq \Delta\mathcal{L}(k)\boldsymbol{x}(k) + \boldsymbol{v}(k)$.

\section{Convergence Analysis}\label{sec:convergence}

In this section, we show that under the proposed protocol~\eqref{eq:protocol}, the system can achieve both mean square average consensus and almost sure consensus.
First, we have the following lemma, which is the cornerstone for the subsequent convergence analysis. 

\begin{lemma}\label{lm:property}
    The system has the following properties:
    \begin{itemize}
        \item[a)] $\mathbf{1}_N^T \bar{\mathcal{L}} = 0$ and $\bar{\mathcal{L}}\mathbf{1}_N = 0$.
        \item[b)] $\bar{\mathcal{L}}$ is a symmetric matrix with $N$ real eigenvalues
        \begin{equation}\label{eq:eigenvalue}
            0=\lambda_1(\bar{\mathcal{L}}) < \lambda_2(\bar{\mathcal{L}}) \leq \cdots \leq \lambda_N(\bar{\mathcal{L}}),
        \end{equation}
        and 
        \begin{equation}\label{eq:lambda2}
            \lambda_2(\bar{\mathcal{L}})=\min_{ x \ne 0, \mathbf{1}_N^T x = 0}\frac{x^T\bar{\mathcal{L}}x}{\Vert x\Vert^2}.
        \end{equation} 
        \item[c)] $\mathbb{E}[\Delta\mathcal{L}(k)] = 0$. 
        \item[d)] $\mathbb{E}[\boldsymbol{v}(k)|\mathcal{F}_k]=0$
        \item[e)] $\mathbb{E}[{x}_i(k)]<\infty, \forall i\in\mathcal{V}, k\geq 0$, 
        $\mathbb{E}\left[ \Vert \boldsymbol{x}(k)\Vert^2 \right]<\infty, \forall k\geq 0$, and $\mathbb{E}[x_i(k)x_j(l)]<\infty, \forall i,j\in\mathcal{V}, k,l\geq 0$.
    \end{itemize}
\end{lemma}
\begin{proof}

    See Appendix~\ref{app:lm_property}.
\end{proof}

% \subsection{Convergence Analysis}
% To show the convergence of the proposed protocol, we first present some assumptions and useful lemmas.

% \begin{assumption}\label{asm:compliant}
%     If $(j,i)\in\mathcal{E}(k)$, then there is a constant $\gamma_h>0$ such that $|h_{ij}(k)|>\gamma_h$.
% \end{assumption}
% \begin{remark}
%     Note that Assumption~\ref{asm:connectivity} and Assumption~\ref{asm:compliant} do not conflict with Assumption~\ref{asm:transmission} a). If the random variable $h_{ij}(k)$ takes the value $0$, then we can say $(j,i)\notin \mathcal{E}(k)$. But overall Assumption~\ref{asm:connectivity} and Assumption~\ref{asm:compliant} should be satisfied at the same time.
% \end{remark}

% \begin{assumption}\label{asm:stepsize2}
%     $\alpha(k+1)\leq \alpha(k), \forall k\geq 0$ and $\limsup_{k\rightarrow\infty}\frac{\alpha(k)}{\alpha(k+1)}<\infty$.
% \end{assumption}

Define 
\begin{equation}
    V(k) \triangleq \Vert(I_N-J)\boldsymbol{x}(k)\Vert^2,
\end{equation} 
where $J=\frac{1}{N}\mathbf{1}_N\mathbf{1}_N^T$. Then we have the following theorem.
\begin{theorem}\label{thm:weak}
    Suppose Assumption~\ref{asm:stepsize} holds, then 
    \begin{equation}\label{eq:lyapunov1}
        \lim_{k\rightarrow\infty}\mathbb{E}[V(k)] = 0.
    \end{equation}
    That is, the multi-agent system can achieve weak consensus.
\end{theorem}
\begin{proof}
    See Appendix~\ref{app:thm_weak}.
\end{proof}

\begin{theorem}\label{thm:ms}
    Suppose Assumption~\ref{asm:stepsize} holds, then under the protocol~\eqref{eq:protocol}, the multi-agent system can achieve mean square average consensus, \textcolor{black}{i.e., $\lim_{k\rightarrow\infty} \mathbb{E}[(x_i(k)-x^*)^2]=0, \forall i\in\mathcal{V}$, where
    \begin{align*}
        \mathbb{E}[x^*]&=\frac{1}{N}\sum_{i=1}^N x_i(0),\\
        Var(x^*)&\leq \frac{\overline{M}_1}{N}\sum_{t=0}^\infty \alpha^2(t),
    \end{align*}
    and $\overline{M}_1$ is defined in~\eqref{eq:M_1}.
    }
\end{theorem}
\begin{proof}
    See Appendix~\ref{app:thm_ms}.
\end{proof} 

\begin{theorem}\label{thm:almost_sure}
    Suppose Assumption~\ref{asm:stepsize} holds, then under the protocol~\eqref{eq:protocol}, the multi-agent system can achieve almost sure consensus.
\end{theorem}
\begin{proof}
    See Appendix~\ref{app:thm_as}.
\end{proof}

% \subsection{Agent-dependent Stepsizes}
\textcolor{black}{
   Note that protocol~\eqref{eq:protocol} requests that each agent uses the same stepsize $\alpha(k)$, which requires agents to at least have consensus on the initial stepsize and the decay rate. For example, if $\alpha(k) = \frac{1}{\max_{i\in\mathcal{V}}\sum_{j\in\mathcal{V}}\bar{a}_{ij}k^p}, \forall k\geq 0$, agents should know $\max_{i\in\mathcal{V}}\sum_{j\in\mathcal{V}}\bar{a}_{ij}$ and $p$ initially. 
   This can be achieved via some coordination among agents, e.g., by running a max-consensus algorithm or relying on the assistance of a central server. However, agents may have different stepsizes, which intrigues to explore the protocol
\begin{equation}\label{eq:protocol2}
       \boldsymbol{x}(k+1) = \left(I_N - A(k)\bar{\mathcal{L}} \right) \boldsymbol{x}(k) + A(k)\boldsymbol{w}(k),
\end{equation}
where $A(k)\triangleq  \mathrm{diag}\left(\alpha_1(k), \alpha_2(k), \ldots, \alpha_N(k)\right)$, and $\alpha_i(k)$ is the stepsize of agent $i$. In the following corollary, it is shown that under protocol~\eqref{eq:protocol2} the system can achieve weak consensus.
\begin{corollary}\label{coro:weak}
    Suppose the stepsizes satisfy
    \begin{itemize}
        \item[a)] $\sum_{k=0}^\infty \alpha_i(k)=\infty$ and $ \sum_{k=0}^\infty \alpha_i^2(k)<\infty, \forall i\in\mathcal{V}$.
        \item[b)] $1-\alpha_i(k) \sum_{j\in\mathcal{V}}\bar{a}_{ij}>0, \forall i\in\mathcal{V}, k\geq 0.$
        \item[c)] $\max_{i,j\in\mathcal{V}}\left|\alpha_i(k)-\alpha_j(k)\right| = o\left(\sum_{i\in\mathcal{V}}\alpha_i(k)\right), k\rightarrow\infty$.
    \end{itemize}
    Then under protocol~\eqref{eq:protocol2}, it holds that
    \begin{equation}\label{eq:weak2}
        \lim_{k\rightarrow\infty}\mathbb{E}[V(k)] = 0.
    \end{equation}
    That is, the multi-agent system can achieve weak consensus.
\end{corollary}
\begin{proof}
    See Appendix~\ref{app:coro1}.
\end{proof}
}
\textcolor{black}{
\begin{corollary}\label{coro:ms}
        Suppose the stepsizes satisfy
    \begin{itemize}
        \item[a)] $\sum_{k=0}^\infty \alpha_i(k)=\infty$ and $ \sum_{k=0}^\infty \alpha_i^2(k)<\infty, \forall i\in\mathcal{V}$.
        \item[b)] $1-\alpha_i(k) \sum_{j\in\mathcal{V}}\bar{a}_{ij}>0, \forall i\in\mathcal{V}, k\geq 0.$
        \item[c)] $\max_{i,j\in\mathcal{V}}\left|\alpha_i(k)-\alpha_j(k)\right| = o\left(\sum_{i\in\mathcal{V}}\alpha_i(k)\right), k\rightarrow\infty$.
        % \item[d)] $\sum_{k=0}^\infty\max_{i,j\in\mathcal{V}}\left|\alpha_i(k)-\alpha_j(k)\right| < \infty$.
    \end{itemize}
    then under protocol~\eqref{eq:protocol2}, the multi-agent system can achieve mean square consensus, i.e., $\lim_{k\rightarrow\infty} \mathbb{E}[(x_i(k)-x^*)^2]=0, \forall i\in\mathcal{V}$, where
    \begin{align*}
        \mathbb{E}[x^*]&=\frac{\mathbf{1}^T}{N}\prod_{t=0}^\infty \left(I_N - A(t)\bar{\mathcal{L}} \right) \boldsymbol{x}(0),\\
        Var(x^*)&\leq\frac{\overline{M}_1}{N}\sup_{k,t\geq 0}\left\Vert\prod_{t}^k \left(I_N - A(t)\bar{\mathcal{L}}\right) \right\Vert^2 \max_{i\in\mathcal{V}}\sum_{t=0}^\infty\alpha_i^2(t).
    \end{align*}
\end{corollary}
\begin{proof}
    See Appendix~\ref{app:coro2}.
\end{proof} 
\begin{corollary}\label{coro:as}
        Suppose the stepsizes satisfy
    \begin{itemize}
        \item[a)] $\sum_{k=0}^\infty \alpha_i(k)=\infty$ and $ \sum_{k=0}^\infty \alpha_i^2(k)<\infty, \forall i\in\mathcal{V}$.
        \item[b)] $1-\alpha_i(k) \sum_{j\in\mathcal{V}}\bar{a}_{ij}>0, \forall i\in\mathcal{V}, k\geq 0.$
        % \item[c)] $\max_{i,j\in\mathcal{V}}\left|\alpha_i(k)-\alpha_j(k)\right| = o\left(\sum_{i\in\mathcal{V}}\alpha_i(k)\right), k\rightarrow\infty$.
        \item[c)] $\sum_{k=0}^\infty\max_{i,j\in\mathcal{V}}\left|\alpha_i(k)-\alpha_j(k)\right| < \infty$.
    \end{itemize}
    then under protocol~\eqref{eq:protocol2}, the multi-agent system can achieve almost sure consensus.
\end{corollary}
\begin{proof}
    See Appendix~\ref{app:coro3}.
\end{proof}
}

\section{Extension to Time-varying Topology}\label{sec:extension}
In a multi-agent system, the network topology may not be strongly connected at each time step due to node failure or link failure. In this section, we will extend the main results to the scenario where the network topology graph is time-varying. 
% In the convergence analysis, only joint connectivity assumption (Assumption~\ref{asm:joint_connectivity}) is required. 

To make the subsequent analysis clearer, we model the network topology as a time-varying undirected graph $\mathcal{G}(k)\triangleq (\mathcal{V}, \mathcal{E}(k))$, where $\mathcal{E}(k)$ denotes the link set at time $k$. If agent $i$ and agent $j$ can communicate with each other at time $k$, we have $(i,j)\in\mathcal{E}(k)$ and $(j,i)\in\mathcal{E}(k)$. Then we have the following definitions.

\begin{definition}[Link failure]
    If the communication link between agent $i$ and agent $j$ suffers from link failure at time $k$, agent $i$ and agent $j$ will be not able to communicate with each other at time $k$, i.e., $(i,j)\notin\mathcal{E}(k)$ and $(j,i)\notin\mathcal{E}(k)$.
\end{definition}
\begin{definition}[Node failure]
    If agent $i$ suffers from node failure at time $k$, it will be not able to communicate with all the other agents at time $k$, i.e., $(i,j)\notin\mathcal{E}(k)$ and $(j, i)\notin\mathcal{E}(k), \forall j\in\mathcal{V}$. Moreover, the failed agent will not update its information state at time $k$, i.e, $a_{ij}(k) = 0$, $\bar{a}_{ij}(k) = 0$, and $v_i(k) = 0, \forall j\in\mathcal{V}$.
\end{definition}
\begin{remark}
    Here $(i,j)\in\mathcal{E}(k)$ only means agent $i$ and agent $j$ have the ability to communicate with each other. They will transmit signals to each other at time $k$ if and only if $(i,j)\in\mathcal{E}(k)$ and $\gamma_i(k) = 1-\gamma_j(k)$, i.e., the link between them does not fail and they choose different time slots to transmit. 
\end{remark}

\begin{assumption}[$L$-Connectivity] \label{asm:joint_connectivity}
We assume for all $k\geq 0$, there exist an integer $L>0$ such that the joint graph $\mathcal{G}(k)\cup\mathcal{G}(k+1)\cup\cdots\cup\mathcal{G}(k+L-1)$ is connected. 
% strongly connected.
\end{assumption}

% \begin{assumption}[Stochastic Assumption]\label{asm:stochastic}
%     The matrix sequence $\{\mathcal{L}(k)\}_{k\geq 0}$ under the designed protocol is adapted to a filtration $\{\mathcal{F}_k\}_{k\geq 0}$ and satisfies  $\mathbf{1}_N^T \mathbb{E}[ \mathcal{L}(k)| \mathcal{F}_k ] = 0$ and $ \mathbb{E}[ \mathcal{L}(k)| \mathcal{F}_k ] \mathbf{1}_N = 0$ for all $k\geq 0$.
% \end{assumption}
For systems with time-varying network topology, an additional assumption on stepsizes should be satisfied.
\begin{assumption}\label{asm:stepsize2}
    $\alpha(k+1)\leq \alpha(k), \forall k \geq 0$ and $\limsup_{k \rightarrow \infty}\frac{\alpha(k)}{\alpha(k+1)} < \infty$.
\end{assumption}

\begin{remark}
    If $\alpha(k)$ is monotonically decreasing, and there exist constants $p\in(0.5,1]$, $q\geq-1$, $c_1>0$ and $c_2>0$, such that for sufficiently large $k$, $\frac{c_1(\ln(k))^q}{\sum_{j\in\mathcal{V}}\bar{a}_{ij} k^p}\leq \alpha(k) \leq \frac{c_2(\ln(k))^q}{\sum_{j\in\mathcal{V}}\bar{a}_{ij} k^p}$, then Assumption~\ref{asm:stepsize} and Assumption~\ref{asm:stepsize2} hold.
\end{remark}

\textcolor{black}{
% Similar to Lemma~\ref{lm:property}, we have the following lemma.
% \begin{lemma}\label{lm:property2}
%     The system has the following properties:
%     \begin{itemize}
%         \item[a)] $\mathbf{1}_N^T \bar{\mathcal{L}}(k) = 0$ and $\bar{\mathcal{L}}(k)\mathbf{1}_N = 0, \forall k\geq 0$.
%         \item[b)] For any $k\geq 0$, $\bar{\mathcal{L}}_k^L\triangleq\sum_{i=k}^{k+L-1}\bar{\mathcal{L}}(i)$ is a symmetric matrix with $N$ real eigenvalues
%         \begin{equation}\label{eq:eigenvalue}
%             0=\lambda_1(\bar{\mathcal{L}}_k^L) < \lambda_2(\bar{\mathcal{L}}_k^L) \leq \cdots \leq \lambda_N(\bar{\mathcal{L}}_k^L),
%         \end{equation}
%         and 
%         \begin{equation}\label{eq:lambda2}
%             \lambda_2(\bar{\mathcal{L}}_k^L)=\min_{ x \ne 0, \mathbf{1}_N^T x = 0}\frac{x^T\bar{\mathcal{L}}_k^Lx}{\Vert x\Vert^2}.
%         \end{equation} 
%         \item[c)] $\mathbb{E}[\Delta\mathcal{L}(k)] = 0$. 
%         \item[d)] $\mathbb{E}[\boldsymbol{v}(k)|\mathcal{F}_k]=0$
%         \item[e)] $\mathbb{E}[{x}_i(k)]<\infty, \forall i\in\mathcal{V}, k\geq 0$, 
%         $\mathbb{E}\left[ \Vert \boldsymbol{x}(k)\Vert^2 \right]<\infty, \forall k\geq 0$, and $\mathbb{E}[x_i(k)x_j(l)]<\infty, \forall i,j\in\mathcal{V}, k,l\geq 0$.
%     \end{itemize}
% \end{lemma}
% \begin{proof}
% The proof of a)~-~d) is similar to Lemma~\ref{lm:property}. Therefore, we omit it and focus on the proof of e).
% \end{proof}
}

\begin{theorem}\label{thm:4.1}
    Suppose Assumptions \ref{asm:stepsize}, \ref{asm:joint_connectivity}, and \ref{asm:stepsize2} hold, then 
    \begin{equation}\label{eq:lyapunov2}
        \lim_{k\rightarrow\infty}\mathbb{E}[V(k)] = 0.
    \end{equation}
\end{theorem}
\begin{proof}
    See Appendix~\ref{app:thm4.1}.
\end{proof}

\begin{theorem}\label{thm:4.2}
    Suppose Assumptions \ref{asm:stepsize}, \ref{asm:joint_connectivity}, and \ref{asm:stepsize2}  hold, then the multi-agent system can achieve mean square average consensus, \textcolor{black}{i.e., $    \lim_{k\rightarrow\infty} \mathbb{E}[(x_i(k)-x^*)^2]=0, \forall i\in\mathcal{V}$, where
    \begin{align*}
        \mathbb{E}[x^*]&=\frac{1}{N}\sum_{i=1}^N x_i(0),\\
        Var(x^*) &\leq \frac{\overline{M}_2}{N}\sum_{t=0}^\infty \alpha^2(t),
    \end{align*}
    and $\overline{M}_2$ is defined in~\eqref{eq:M_2}.
    }
\end{theorem}
\begin{proof}
    See Appendix~\ref{app:thm4.2}.
\end{proof} 

\begin{lemma}\label{lm:as}
    \textcolor{black}{Suppose Assumptions~\ref{asm:stepsize} and \ref{asm:stepsize2} hold. If there exist integers $L_0>0$ and $k_0\geq 0$ such that $\inf_{m\geq 0}\lambda_2\left(\sum_{i=k_0+mL_0}^{k_0+(m+1)L_0-1}\mathcal{L}(i)\right)>0, \forall l\in\{0,1,\ldots, L_0-1\}$, then under the protocol~\eqref{eq:protocol}, it holds that
    \begin{equation}
        \lim_{m\rightarrow\infty} V(k_0+mL_0) = 0 \quad a.s.
    \end{equation}
    }
\end{lemma}
\begin{proof}
    See Appendix~\ref{app:lm_as}.
\end{proof}

\begin{theorem}\label{thm:4.3}
    Suppose Assumptions~\ref{asm:stepsize}, ~\ref{asm:joint_connectivity}, and \ref{asm:stepsize2} hold, then the multi-agent system can achieve almost sure consensus under the protocol~\eqref{eq:protocol}.
\end{theorem}
\begin{proof}
    See Appendix~\ref{app:thm4.3}.
\end{proof}

\section{Simulation Results}
% In this subsection, we consider there are $N=5$ agents in the system, i.e., $\mathcal{V}=\{1,2,\ldots,5\}$. We set the initial information states as $x_i(0)=i, \forall i\in\mathcal{V}$. The fading channels $h_{ij}(k), \forall i,j\in\mathcal{V}, k\geq 0$ are modeled as i.i.d. random processes with $\Lambda_{ij} = 1, \forall i,j\in\mathcal{V}$. Moreover, we set $p_i = 0.5, \forall i\in\mathcal{V}$. Other parameters are provided along with the figures.

% Moreover, we set $\sigma_i^2 = -60dB, \forall i\in\mathcal{V}$ and $p_i = 0.5, \forall i\in\mathcal{V}$.

In this subsection, we first investigate the convergence capability of the protocol proposed in~\cite{molinari2018exploiting} under noisy channels and non-coherent transmission, and then present simulation results to show the effectiveness of our scheme. 

In~\cite{molinari2018exploiting}, agent $i$ needs to transmit two messages, i.e., the current information state $x_i(k)$ and a constant $u_i(k)\equiv 1$, via the same wireless channel at time step $k$. After receiving, agent $i$ can obtain two values as follows
\begin{align}
    y_i^\mathrm{(1)}(k) &= \left\Vert \sum_{j\in\mathcal{V}}  h_{ij}(k)\sqrt{x_j(k)} + n_i(k) \right\Vert^2,\\
     y_i^\mathrm{(2)}(k) &= \left\Vert \sum_{j\in\mathcal{V}}  h_{ij}(k) + n_i(k) \right\Vert^2.
\end{align}
Then, agent $i$ updates its information state as follows
\begin{equation}\label{eq:comparison}
    x_i(k+1) = \frac{y_i^\mathrm{(1)}(k)}{y_i^\mathrm{(2)}(k)}.
\end{equation}

We consider the physical network topology is a time-invariant complete graph, and there are $N=5$ agents in the system, i.e., $\mathcal{V}=\{1,2,\ldots,5\}$. We set the initial information states as $x_i(0)=i, \forall i\in\mathcal{V}$. The fading channels $h_{ij}(k), \forall i,j\in\mathcal{V}, k\geq 0$ are modeled as i.i.d. random processes with $\Lambda_{ij} = 1, \forall i,j\in\mathcal{V}$, and the channel noises $n_i(k), \forall i\in\mathcal{V}, k\geq 0$ are i.i.d. random processes with $\sigma_i^2=-60dB, \forall i\in\mathcal{V}$. 
Fig.~\ref{fig:comparison} shows that the information state of each agent diverges under the protocol~\eqref{eq:comparison}, which is vulnerable to channel noises and asynchronization of transmitters, even thought under such a simple setting.

\begin{figure}[h]
	\centering
	\includegraphics[width=3in]{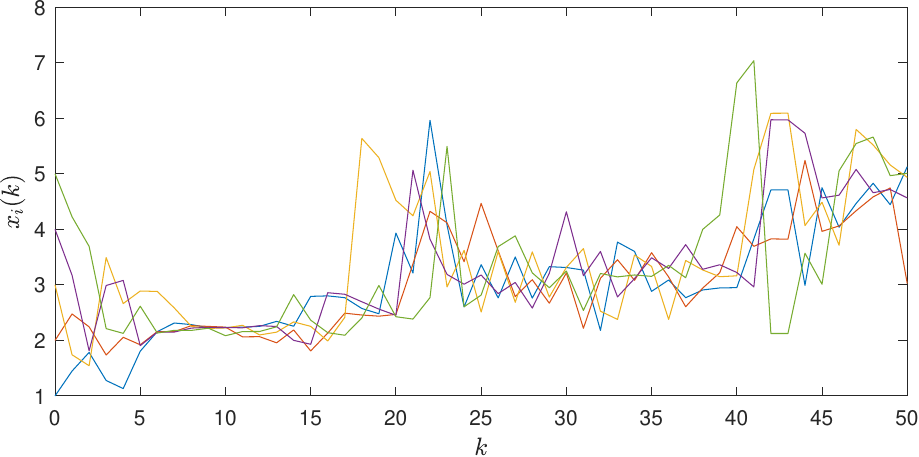}
	\caption{Evolution of information states $x_i(k), \forall i \in \mathcal{V}$ under the protocol~\eqref{eq:comparison}~\cite{molinari2018exploiting} with $\sigma^2_i = -60dB, \forall i\in\mathcal{V}$.}
	\label{fig:comparison}
\end{figure}

% Fig.~\ref{fig:proposed} shows the evolution of information states under the proposed protocol~\eqref{eq:protocol}. It can be seen that the system can achieve average consensus asymptotically under different channel noises. The variance of $x^*$ is determined by the noise power $\sigma_i^2$. When $\sigma_i^2$ becomes larger, the variance of $x^*$ also becomes larger.

\begin{figure}[h]
	\centering
	\includegraphics[width=0.9\linewidth]{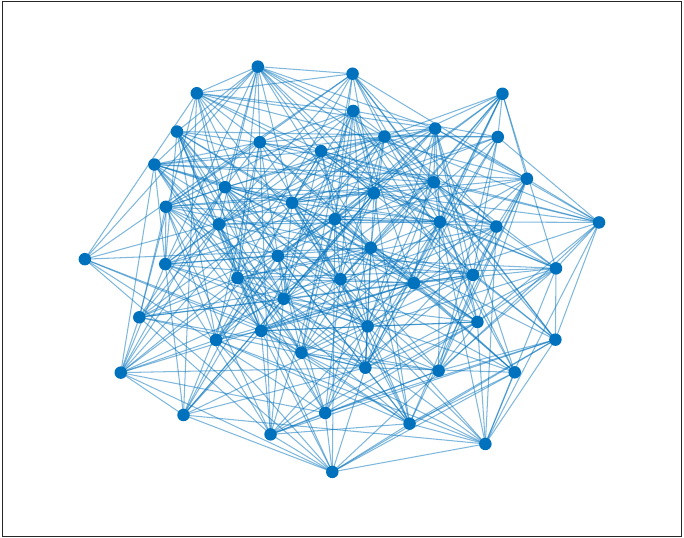}
	\caption{\textcolor{black}{A network topology with 50 agents.}}
	\label{fig:graph}
\end{figure}

\textcolor{black}{In the following, we will present the simulation results of a system with time-varying topology. There are $50$ agents in this system. The network topology is a $L$-connected time-varying graph with $L=3$. Specifically, the time-varying graph is generated by sampling from the network topology presented in Fig.~\ref{fig:graph}.
At time steps $k\ne nL-1, n\in\{1,2,\ldots\}$, we randomly pick each agent with probability $q = 0.6$. At time steps $k = nL-1, n\in\{1,2,\ldots\}$, we pick the agents which were not picked in the previous $L-1$ iterations, and also at least one of the other agents. 
% The picked agents can communicate with their neighbors which are picked and choose the same time slot to transmit.
The initial information states are uniformly selected from $[0,100]$. Moreover, we set $p_i = 0.5, \forall i\in\mathcal{V}$. Other parameters are provided along with the figures.
}

% Moreover, we set $\sigma_i^2 = -60dB, \forall i\in\mathcal{V}$ and $p_i = 0.5, \forall i\in\mathcal{V}$.} 

\textcolor{black}{Fig.~\ref{fig:noise} shows the convergence behavior of the system under different noise powers. It can be seen that the system can achieve average consensus asymptotically, and the variance of $x^*$ is influenced by the noise power $\sigma_i^2$. Generally, the higher the noise power is, the worse the convergence performance is. However, when the noise power is small enough, the effect of channel fading becomes the major factor that influences the convergence performance.
Then, Fig.~\ref{fig:fading} presents the convergence behavior of the system under different channel fading, which shows that a larger $\Lambda_{ij}$ leads to a higher variance. 
}

\begin{figure}[h]
	\centering
	\includegraphics[width=\linewidth]{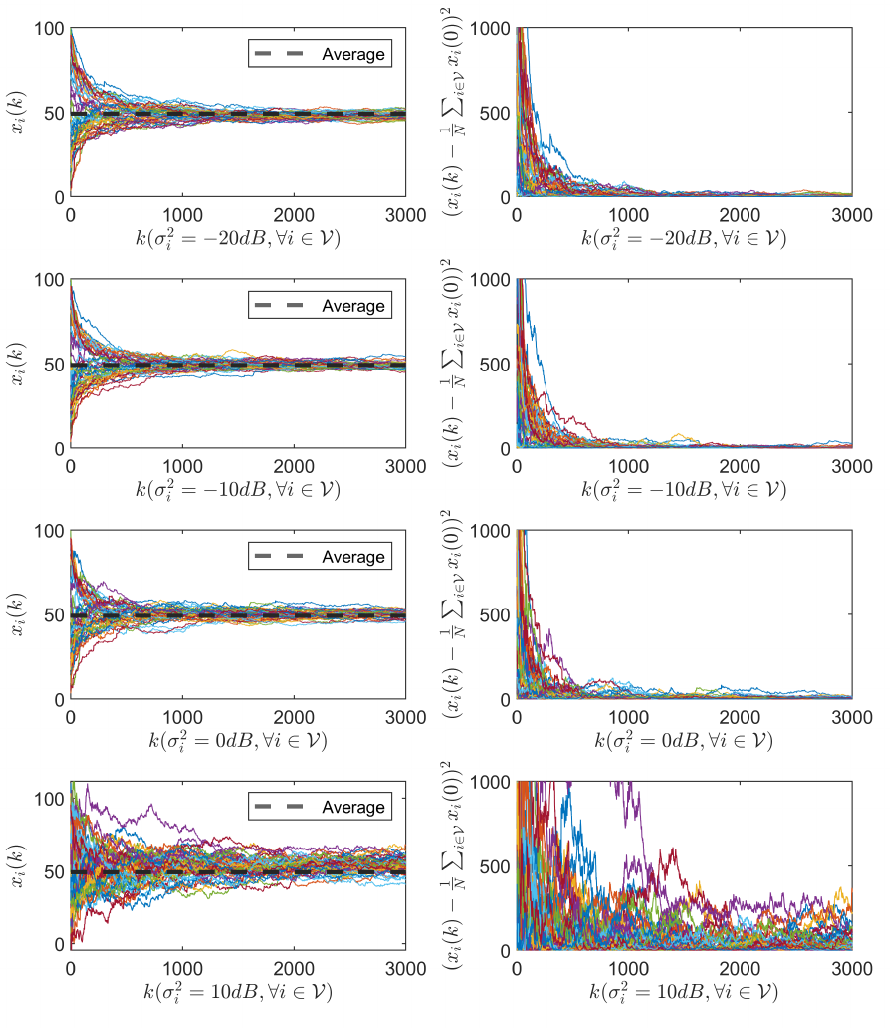}
	\caption{\textcolor{black}{Evolution of information state $x_i(k)$ and the mean square error $\left(x_i(k)-\frac{1}{N}\sum_{i\in\mathcal{V}}x_i(0)\right)^2, \forall i \in \mathcal{V}$ under different noise power ($\Lambda_{ij}=2, \forall i,j\in\mathcal{V}$).}}
	\label{fig:noise}
\end{figure}

\begin{figure}[h]
	\centering
	\includegraphics[width=\linewidth]{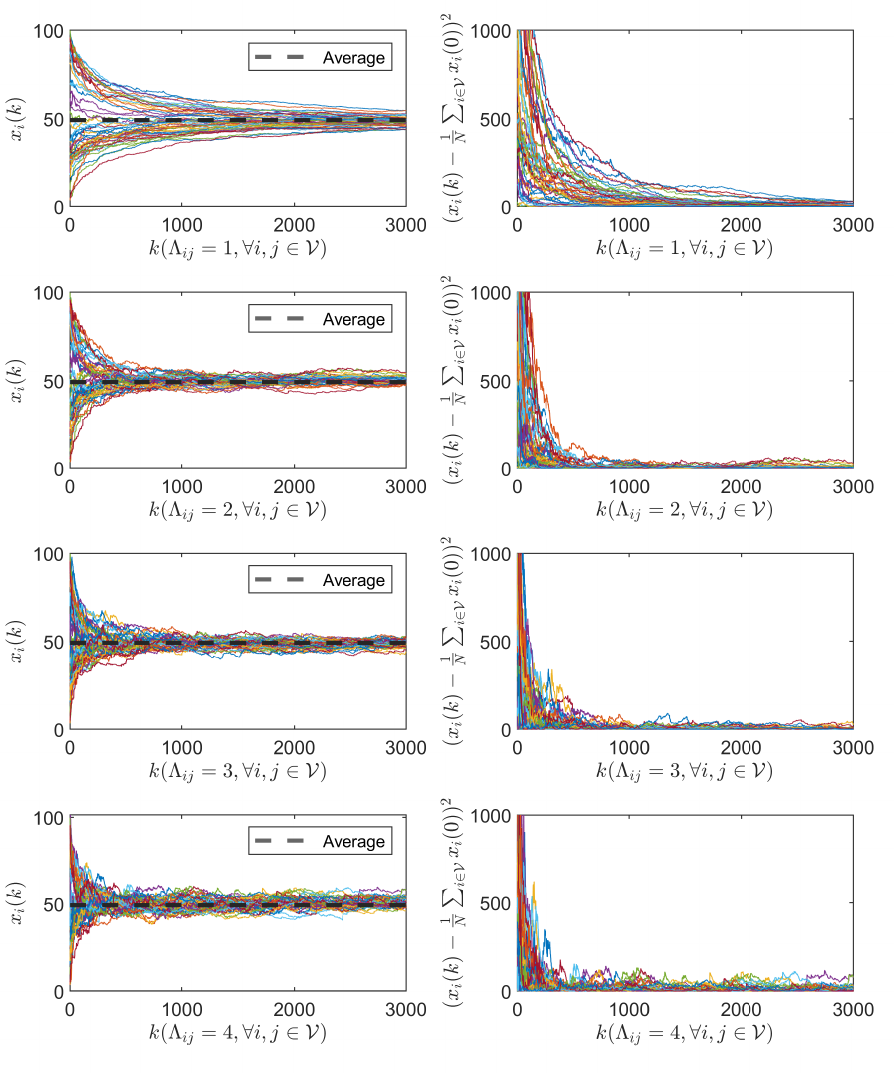}
	\caption{\textcolor{black}{Evolution of information state $x_i(k)$ and the mean square error $\left(x_i(k)-\frac{1}{N}\sum_{i\in\mathcal{V}}x_i(0)\right)^2, \forall i \in \mathcal{V}$ under different channel fading ($\sigma_i^2 = 0dB, \forall i \in \mathcal{V}$).}}
	\label{fig:fading}
\end{figure}

\section{Conclusion}
In this paper, we consider the distributed average consensus problem in a multi-agent system, where over-the-air aggregation is adopted to save communication resources. However, over-the-air aggregation is vulnerable to channel noises and non-coherent transmission. To handle the effect caused by noises and non-coherent transmission, we propose a stochastic approximate based protocol, under which the system can achieve mean square average consensus and almost sure consensus. Furthermore, we extend the analysis to the scenarios with time-varying network topology. Numerical simulations show the effectiveness of the proposed protocol.

\bibliographystyle{unsrt}
\bibliography{ref.bib}

\newpage
\appendices
\section{Proof of Lemma~\ref{lm:property}}\label{app:lm_property}
\begin{itemize}
        \item[a)] By $\Lambda_{ij} = \Lambda_{ji}$, we have $\bar{a}_{ij} = \bar{a}_{ji}$. Due to $\mathbb{E}[\mathcal{A}(k)] = [\bar{a}_{ij}]$ and the definition of $\mathcal{D}(k)$, $\mathbf{1}_N^T \bar{\mathcal{L}} = 0$ and $\bar{\mathcal{L}}\mathbf{1}_N = 0$ hold.
        \item[b)] First, it is easy to see $\bar{\mathcal{L}}$ is symmetric. By a) we can regard $\bar{\mathcal{L}}$ as the Laplacian matrix of a \textcolor{black}{connected} weighted digraph. Therefore, \eqref{eq:eigenvalue} and \eqref{eq:lambda2} hold~\cite{huang2010stochastic}.
        \item[c)] By the definitions of $\Delta \mathcal{L}(k)$ and $\bar{\mathcal{L}}$, we have $$\mathbb{E}[\Delta\mathcal{L}(k)] = \mathbb{E}[\bar{\mathcal{L}} - \mathcal{L}(k)]= \bar{\mathcal{L}} -\mathbb{E}[\mathcal{L}(k)] = \bar{\mathcal{L}} - \bar{\mathcal{L}} = 0.
        $$

        \item[d)]\textcolor{black}{
        By~\eqref{eq:e_Y1}~-~\eqref{eq:e_Y2} and the definition of $n_i(k)$, $\mathbb{E}[\boldsymbol{v}(k)|\mathcal{F}_k]=0$ holds.
        }

        \item[e)]\textcolor{black}{
        We prove $\mathbb{E}[{x}_i(k)]<\infty, \forall i\in\mathcal{V}, k\geq 0$ and $\mathbb{E}\left[\Vert \boldsymbol{x}(k) \Vert^2\right]<\infty, \forall k\geq 0$ by induction. 
        It has been known that $\mathbb{E}[{x}_i(0)]=x_i(0)<\infty, \forall i\in\mathcal{V}$. Moreover, if $\mathbb{E}[{x}_i(k)]<\infty, \forall i\in\mathcal{V}$, we have
        \begin{equation}
            \begin{split}
               \mathbb{E}[x_i(k+1)] =& \left(1 - \alpha(k) \sum_{j\in\mathcal{V}}\bar{a}_{ij} \right) \mathbb{E}[x_i(k)] \\& + \alpha(k) \left(\mathbb{E}\left[|y_i(k)|^2\right] -\sigma_i^2\right)
               \\=&\left(1 - \alpha(k) \sum_{j\in\mathcal{V}}\bar{a}_{ij} \right) \mathbb{E}[x_i(k)] \\& + \alpha(k) \sum_{j\in\mathcal{V}} \bar{a}_{ij} \mathbb{E}\left[x_j(k)\right]
               \\\leq& \max_{v\in\mathcal{V}} \mathbb{E}\left[x_v(k)\right]<\infty.
            \end{split}
        \end{equation}
        Therefore, $\mathbb{E}[{x}_i(k)]<\infty, \forall i\in\mathcal{V}, \forall k\geq 0$.
\\
        % Then we use conditioning to show $\mathbb{E}\left[\Vert \boldsymbol{x}(k) \Vert^2\right]<\infty, \forall k\geq 0$. 
        % It holds that
        % \begin{equation}
        % \begin{split}
        %     &\mathbb{E}\left[\Vert \boldsymbol{x}(k+1) \Vert^2 \left.\right| \mathcal{F}_k \right]\\&=\mathbb{E}\left[\Vert\left(I_N - \alpha(k)\mathcal{L}(k) \right) \boldsymbol{x}(k) + \alpha(k)\boldsymbol{v}(k)\Vert^2\right]
        %     \\&\leq \mathbb{E}\left[ \Vert\left(I_N - \alpha(k)\mathcal{L}(k) \right)\Vert^2\right] \Vert \boldsymbol{x}(k)\Vert^2 
        %     \\ &\quad + \alpha^2(k)\mathbb{E}\left[ \left. \Vert\boldsymbol{v}(k)\Vert^2 \right| \mathcal{F}_k\right]
        %     \\ &\quad + 2\alpha(k)\boldsymbol{x}^T(k)\mathbb{E}\left[ \left. \boldsymbol{v}(k) \right| \mathcal{F}_k\right].
        % \end{split}
        % \end{equation}
        % Combining with $\Vert \boldsymbol{x}(0) \Vert^2 <\infty$ and Lemma~\ref{lm:conditioning}, $\mathbb{E}\left[\Vert \boldsymbol{x}(k) \Vert^2\right]<\infty, \forall k\geq 0$ can be proved.
        % \begin{lemma}\citep[Theorem 4.1.13]{durrett2019probability}\label{lm:conditioning}
        %     Suppose $\sigma_1(X)$ and $\sigma_2(X)$ are two $\sigma$-fields generated by $X$. If $\sigma_2(X) \subset \sigma_2(X)$, $\mathbb{E}\left[\mathbb{E}[X|\sigma_2(X)]|\sigma_1(X)\right] = \mathbb{E}[X|\sigma_1(X)]$.
        % \end{lemma}
\\
        To prove $\mathbb{E}\left[\Vert \boldsymbol{x}(k) \Vert^2\right]<\infty, \forall k\geq 0$, the following lemma is needed.
        \begin{lemma}[Cauchy–Schwarz inequality]\label{lm:cs_inequality}
            \textcolor{black}{$$\left|\mathbb{E}[XY]\right|^2\leq \mathbb{E}[X^2] \mathbb{E}[Y^2],$$} where $X$ and $Y$ are random variables.
        \end{lemma}
        By the definitions of $\Gamma_{ij}(k)$, $h_{ij}(k)$ and $u$, we have
        \begin{equation}\label{eq:Gamma_square}
        \begin{split}
            \mathbb{E}\left[\Gamma_{ij}^2(k)\right] 
            = & \mathbb{E}\left[(\gamma_i(k)(1-\gamma_j(k))+ \gamma_j(k)(1-\gamma_i(k)))^2\right]\\
            = & \mathbb{E}\left[ \gamma_i(k)^2(1-\gamma_j(k))^2 +\gamma_j^2(k)(1-\gamma_i(k))^2 \right.\\ 
            & \left. + 2\gamma_i(k)\gamma_j(k)(1-\gamma_i(k))(1-\gamma_j(k))\right]\\
            = & \mathbb{E}\left[\gamma_i^2(k)\right]\mathbb{E}\left[(1-\gamma_j(k))^2\right] \\ 
            & + \mathbb{E}\left[\gamma_j^2(k)\right]\mathbb{E}\left[(1-\gamma_i(k))^2\right]\\
            & + 2\mathbb{E}\left[\gamma_i(k)(1-\gamma_i(k)\right]\mathbb{E}\left[\gamma_j(k)(1-\gamma_j(k))\right] \\
            = & p_i(1-p_j)+p_j(1-p_i) < \infty,
        \end{split}
        \end{equation}
        and
        \begin{equation}\label{eq:re_im}
        \begin{split}
            &\mathbb{E}\left[\left(\mathrm{Re}\left[h_{ij}(k)u\right]\mathrm{Re}\left[h_{il}(k)u\right] \right.\right.\\&\quad \left.\left.+\mathrm{Im}\left[h_{ij}(k)u\right] \mathrm{Im}\left[h_{il}(k)u\right]\right)^2 \right]
            \\&= \mathbb{E}\left[\left(\mathrm{Re}\left[h_{ij}(k)u\right]\right)^2\right]\mathbb{E}\left[\left(\mathrm{Re}\left[h_{il}(k)u\right]\right)^2\right]\\&\quad+\mathbb{E}\left[\left(\mathrm{Im}\left[h_{ij}(k)u\right] \right)^2\right]\mathbb{E}\left[\left(\mathrm{Im}\left[h_{il}(k)u\right]\right)^2 \right]\\&\quad+2\mathbb{E}\left[\mathrm{Re}\left[h_{ij}(k)u\right]\mathrm{Im}\left[h_{ij}(k)u\right]\right]\\&\quad\times\mathbb{E}\left[ \mathrm{Re}\left[h_{il}(k)u\right]\mathrm{Im}\left[h_{il}(k)u\right]\right]
            \\&=\mathbb{E}\left[\left(\mathrm{Re}[h_{ij}(k)] \mathrm{Re}[u]-\mathrm{Im}[h_{ij}(k)] \mathrm{Im}[u]\right)^2\right]
            \\&\quad\times\mathbb{E}\left[\left(\mathrm{Re}[h_{il}(k)] \mathrm{Re}[u]-\mathrm{Im}[h_{il}(k)] \mathrm{Im}[u]\right)^2\right]
            \\&\quad+\mathbb{E}\left[\left(\mathrm{Re}[h_{ij}(k)] \mathrm{Im}[u]+\mathrm{Im}[h_{ij}(k)] \mathrm{Re}[u]\right)^2\right]
            \\&\quad\times\mathbb{E}\left[\left(\mathrm{Re}[h_{il}(k)] \mathrm{Im}[u]+\mathrm{Im}[h_{il}(k)] \mathrm{Re}[u]\right)^2\right]
            \\&\quad+2\mathbb{E}\left[(\mathrm{Re}[h_{ij}(k)] \mathrm{Re}[u]-\mathrm{Im}[h_{ij}(k)] \mathrm{Im}[u])\right.\\&\left.\quad\times(\mathrm{Re}[h_{ij}(k)] \mathrm{Im}[u]+\mathrm{Im}[h_{ij}(k)] \mathrm{Re}[u])\right]
            \\&\quad\times\mathbb{E}\left[(\mathrm{Re}[h_{il}(k)] \mathrm{Re}[u]-\mathrm{Im}[h_{il}(k)] \mathrm{Im}[u])\right.\\&\left.\quad\times(\mathrm{Re}[h_{il}(k)] \mathrm{Im}[u]+\mathrm{Im}[h_{il}(k)] \mathrm{Re}[u])\right]
            \\&
            =\frac{1}{2}\Lambda_{ij}\Lambda_{il} <\infty,
        \end{split}
        \end{equation}
         By Lemma~\ref{lm:cs_inequality}, it holds if we assume $\mathbb{E}[x_i^2(k)]<\infty$
        \begin{equation}
            \left| \mathbb{E}[x_i(k)x_j(k)]\right|^2< \infty.
        \end{equation}
        Then, by combining \eqref{eq:Gamma_square} and \eqref{eq:re_im}, we have 
        \begin{equation}\label{eq:Y1}
        \begin{split}
            &\mathbb{E}\left[ \left( Y^\mathrm{(1)}_{i,jl}(k) \right)^2 \right]\\&= \rho^2\mathbb{E}\left[ x_j(k) x_l(k) \right] \mathbb{E}\left[\Gamma_{ij}^2(k)\right]\mathbb{E}\left[\Gamma_{il}^2(k)\right] \\&\quad\times\mathbb{E}\left[\left(\mathrm{Re}\left[h_{ij}(k)u\right]\mathrm{Re}\left[h_{il}(k)u\right] \right. \right.\\&\quad  \left.+\mathrm{Im}\left[h_{ij}(k)u\right] \mathrm{Im}\left[h_{il}(k)u\right]\right)^2 \left.\right]
            \\& =\frac{1}{2}\rho^2 \Lambda_{ij}\Lambda_{il}\mathbb{E}\left[ x_j(k) x_l(k) \right] \\&\quad\times\left(p_i(1-p_j)+p_j(1-p_i)\right)\left(p_i(1-p_l)+p_l(1-p_i)\right)
            \\&\leq M^\mathrm{(1)}_{ijl}\max_{v\in\mathcal{V}}\mathbb{E}\left[  x_v^2(k)  \right] 
            \\&<\infty,
        \end{split}
        \end{equation}
        where $M^\mathrm{(1)}_{ijl}\triangleq \frac{1}{2}\rho^2\Lambda_{ij}\Lambda_{il}\left(p_i(1-p_j)+p_j(1-p_i)\right)$ $\times\left(p_i(1-p_l)+p_l(1-p_i)\right)$.
        \\
        % If $\mathbb{E}[x_i(k)]<\infty$ holds, then we have
        Similarly, we have
        \begin{equation}\label{eq:Y2}
        \begin{split}
            &\mathbb{E}\left[\left(Y_{ij}^{(2)}(k)\right)^2\right] \\&= \rho\mathbb{E}\left[ x_j(k)\right]\mathbb{E}\left[\Gamma_{ij}^2(k)\right] \mathbb{E}\left[\left(\mathrm{Re}\left[h_{ij}(k)u\right] \mathrm{Re}[n_i(k)] \right.\right.\\& \quad\left.+ \mathrm{Im}\left[h_{ij}(k)u\right] \mathrm{Im}[n_i(k)]\right)^2\left.\right]
            \\&=\frac{1}{2}\rho\Lambda_{ij}\sigma_i^2 \mathbb{E}\left[ x_j(k) \right] \left(p_i(1-p_j)+p_j(1-p_i)\right)
            \\&\leq M^\mathrm{(2)}_{ij}\max_{v\in\mathcal{V}} x_v(0) 
            \\&<\infty.
        \end{split}
        \end{equation}
        where $M^\mathrm{(2)}_{ij}\triangleq \frac{1}{2}\rho\Lambda_{ij}\sigma_i^2  \left(p_i(1-p_j)+p_j(1-p_i)\right)$.
        According to the definition of $n_i(k)$, we have
        \begin{equation}\label{eq:e1}
        \begin{split}
            &\mathbb{E}\left[ \left(|n_i(k)|^2-\sigma_i^2\right)^2 \right]= 7\sigma_i^4.
        \end{split}
        \end{equation}
        By the independence of $h_{ij}(k), i, j \in\mathcal{V}$ and $n_i(k), i\in\mathcal{V}$, we have
        \begin{align}
        \begin{split}\label{eq:e2}
            &\mathbb{E}\left[ Y^\mathrm{(1)}_{i,jl}(k) Y^\mathrm{(1)}_{i',jl}(k)  \right]= 0
        \end{split}\\
        \begin{split}\label{eq:e3}
            &\mathbb{E}\left[ Y^\mathrm{(1)}_{i,jl}(k) Y^\mathrm{(1)}_{i,j'l}(k)  \right]= 0
        \end{split}\\
        \begin{split}\label{eq:e4}
            &\mathbb{E}\left[ Y_{ij}^{(2)}(k) Y_{i'j}^{(2)}(k) \right]= 0
        \end{split}\\
        \begin{split}\label{eq:e5}
            &\mathbb{E}\left[ Y_{ij}^{(2)}(k) Y_{ij'}^{(2)}(k) \right]= 0
        \end{split}\\
        \begin{split}\label{eq:e6}
            &\mathbb{E}\left[ Y^\mathrm{(1)}_{i,jl}(k)\left(|n_i(k)|^2-\sigma_i^2\right) \right]= 0
        \end{split}\\
        \begin{split}\label{eq:e7}
            &\mathbb{E}\left[ Y_{ij}^{(2)}(k)\left(|n_i(k)|^2-\sigma_i^2\right) \right]= 0
        \end{split}
        \end{align}
        By~\eqref{eq:Y1}~-~\eqref{eq:e7}, we have
        \begin{equation}\label{eq:v_square}
        \begin{split}
            &\mathbb{E}\left[ v_i^2(k) \right] \\& =\mathbb{E}\left[\left(|n_i(k)|^2-\sigma_i^2\right)^2\right] + \mathbb{E}\left[\sum_{j\in\widetilde{\mathcal{N}}_i}\sum_{l\in\widetilde{\mathcal{N}}_i\backslash j}\left(Y_{i,jl}^{(1)}(k)\right)^2\right] \\&\quad + \mathbb{E}\left[2\sum_{j\in\widetilde{\mathcal{N}}_i}\left(Y_{ij}^{(2)}(k)\right)^2\right] 
            \\& \leq 7\sigma_i^4 + \sum_{j\in\widetilde{\mathcal{N}}_i}\sum_{l\in\widetilde{\mathcal{N}}_i\backslash j}M^\mathrm{(1)}_{ijl}\max_{v\in\widetilde{\mathcal{N}}_i}\mathbb{E}\left[ x_v^2(k)\right] \\&\quad +2\sum_{j\in\widetilde{\mathcal{N}}_i}M^\mathrm{(2)}_{ij}\max_{v\in\widetilde{\mathcal{N}}_i} x_v(0)
            \\& \leq 7\sigma_i^4 + \sum_{j\in\widetilde{\mathcal{N}}_i}\sum_{l\in\widetilde{\mathcal{N}}_i\backslash j}M^\mathrm{(1)}_{ijl}\mathbb{E}\left[ \Vert \boldsymbol{x}(k)\Vert^2 \right] \\&\quad +2\sum_{j\in\widetilde{\mathcal{N}}_i}M^\mathrm{(2)}_{ij}\max_{v\in\widetilde{\mathcal{N}}_i} x_v(0).
        \end{split}
        \end{equation}
        By the definition of $n_i(k)$ and \eqref{eq:e_Y1}~-~\eqref{eq:e_Y2}, we have
        \begin{equation}\label{eq:xv}
        \begin{split}
            &\mathbb{E}\left[ x_i(k)v_i(k) \right]= 0.
        \end{split}
        \end{equation}
        By the definition of $\Delta\mathcal{L}(k)$, we have
        \begin{equation}\label{eq:delta_l}
        \begin{split}
            &\mathbb{E}\left[ \Vert \Delta\mathcal{L}(k) \Vert^2\right] \\&\leq \mathbb{E}\left[ \Vert \Delta\mathcal{L}(k) \Vert_F^2\right] 
            \\&= \mathbb{E}\left[ \sum_{i\in\mathcal{V}}\sum_{j\in\widetilde{\mathcal{N}}_i} (a_{ij}(k)-\bar{a}_{ij})^2\right]
            \\&=\sum_{i\in\mathcal{V}}\sum_{j\in\widetilde{\mathcal{N}}_i} \left(\mathbb{E}\left[ a_{ij}^2(k)\right] -\bar{a}_{ij}^2\right) 
            \\&= \sum_{i\in\mathcal{V}}\sum_{j\in\widetilde{\mathcal{N}}_i} \left(\mathbb{E}\left[ \rho^2 \Gamma^2_{ij}(k)|h_{ij}(k)|^4\right] -\bar{a}_{ij}^2\right)
            \\&= C_L,
        \end{split}
        \end{equation}
        where $C_L\triangleq \sum_{i\in\mathcal{V}}\sum_{j\in\widetilde{\mathcal{N}}_i} \left( 8\Lambda_{ij}^2\rho^2\left(p_i(1-p_j)\right.\right.$ $\left.\left.+p_j(1-p_i)\right)-\bar{a}_{ij}^2\right)$. 
        Then, by~\eqref{eq:v_square}~-~\eqref{eq:delta_l}, we have
        \begin{equation}\label{eq:w_square}
        \begin{split}
           \mathbb{E}\left[ \Vert \boldsymbol{w}(k) \Vert^2 \right] &= \mathbb{E}\left[ \Vert\Delta\mathcal{L}(k)\boldsymbol{x}(k) + \boldsymbol{v}(k)\Vert^2\right]
           \\& = \mathbb{E}\left[ \Vert\Delta\mathcal{L}(k)\boldsymbol{x}(k) \Vert^2\right] + \mathbb{E}\left[\Vert\boldsymbol{v}(k)\Vert^2\right]
           \\&\leq \mathbb{E}\left[ \Vert\Delta\mathcal{L}(k)\Vert^2 \right] \mathbb{E}\left[\Vert\boldsymbol{x}(k) \Vert^2\right] \\&\quad + \mathbb{E}\left[\Vert\boldsymbol{v}(k)\Vert^2\right]
        \end{split}
        \end{equation}
        By the assumption $\mathbb{E}[x_i^2(k)]<\infty$ and~\eqref{eq:w_square}, we have
        \begin{equation}\label{eq:E_x_square}
        \begin{split}
             &\mathbb{E}\left[\Vert \boldsymbol{x}(k+1)\Vert^2\right] \\&= \mathbb{E}\left[ \Vert \left(I_N - \alpha(k)\bar{\mathcal{L}} \right) \boldsymbol{x}(k) + \alpha(k)\boldsymbol{w}(k)\Vert^2\right]
             \\& = \mathbb{E}\left[ \Vert \left(I_N - \alpha(k)\bar{\mathcal{L}} \right) \boldsymbol{x}(k) \Vert^2 \right] + \mathbb{E}\left[ \Vert \alpha(k)\boldsymbol{w}(k)\Vert^2\right]
             \\&\leq \Vert I_N - \alpha(k)\bar{\mathcal{L}}  \Vert^2 \mathbb{E}\left[ \Vert\boldsymbol{x}(k) \Vert^2\right]  +  \alpha^2(k) \mathbb{E}\left[ \Vert \boldsymbol{w}(k)\Vert^2\right]
             \\&\leq  \left(1-2\alpha(k)\lambda_2\left(\bar{\mathcal{L}}\right)+\alpha^2(k)\left\Vert \bar{\mathcal{L}}\right\Vert^2\right)\mathbb{E}\left[ \Vert\boldsymbol{x}(k) \Vert^2\right] \\&\quad+ \alpha^2(k)\mathbb{E}\left[ \Vert \boldsymbol{w}(k)\Vert^2\right] 
             \\& \leq \left(1-2\alpha(k)\lambda_2\left(\bar{\mathcal{L}}\right)+\alpha^2(k)\left\Vert \bar{\mathcal{L}}\right\Vert^2\right.\\&\left.\quad+\alpha^2(k)\left(C_L+C_M^\mathrm{(1)}\right) \right) \mathbb{E}\left[ \Vert\boldsymbol{x}(k) \Vert^2\right] +\alpha^2(k)C_M^\mathrm{(2)}
             \\&\leq \Phi_1(k:0)\Vert\boldsymbol{x}(0)\Vert^2 + \sum_{t=0}^k\Phi_1(k:t+1)\alpha^2(t)C_M^\mathrm{(2)},
        \end{split}
        \end{equation}
        where 
        \begin{align*}
            &C_M^\mathrm{(1)}\triangleq\sum_{j\in\mathcal{V}}\sum_{l\in\mathcal{V}\backslash j}M^\mathrm{(1)}_{ijl},\\
            &C_M^\mathrm{(2)}\triangleq \sum_{i\in\mathcal{V}}\left( 7\sigma_i^4 +2\sum_{j\in\mathcal{V}}M^\mathrm{(2)}_{ij}\max_{v\in\mathcal{V}} x_v(0) \right),\\
            &\Phi_1(k:l)\triangleq \prod_{t=k}^l \left(1-2\alpha(k)\lambda_2\left(\bar{\mathcal{L}}\right)+\alpha^2(k)\left\Vert \bar{\mathcal{L}}\right\Vert^2\right.\\& \qquad\qquad\left.\quad+\alpha^2(k)\left(C_L+C_M^\mathrm{(1)}\right) \right).
        \end{align*}
        % As a result, by \eqref{eq:Y1}, \eqref{eq:Y2}, and Lemma~\ref{lm:cs_inequality}, it holds when $\mathbb{E}[x_i(k)]<\infty, \forall i\in\mathcal{V}$ and $\mathbb{E}[\Vert \boldsymbol{x}(k) \Vert^2]<\infty$
        % \begin{equation}
        %     \mathbb{E}[\Vert \boldsymbol{x}(k+1) \Vert^2]<\infty.
        % \end{equation}
        Furthermore, due to $\Vert \boldsymbol{x}(0) \Vert^2 < \infty$, $\mathbb{E}\left[\Vert \boldsymbol{x}(k) \Vert^2\right]<\infty, \forall k\geq 0$ holds. Again, by Lemma~\ref{lm:cs_inequality}, we have $\mathbb{E}[x_i(k)x_j(l)]<\infty, \forall i,j\in\mathcal{V}, k,l\geq 0$.}
    \end{itemize}

\section{Proof of Theorem~\ref{thm:weak}}\label{app:thm_weak}
Define $\boldsymbol{\delta}(k) \triangleq (I_N-J)\boldsymbol{x}(k)$, then $V(k)=\boldsymbol{\delta}^T(k)\boldsymbol{\delta}(k)$. Since $\bar{\mathcal{L}} J=0$ and $J \bar{\mathcal{L}} = 0$, we have 
    \begin{equation}\label{eq:delta}
    \begin{split}
       \boldsymbol{\delta}(k+1) &= (I_N - J) \left((I_N-\alpha(k)\bar{\mathcal{L}})\boldsymbol{x}(k)+\alpha(k)\boldsymbol{w}(k)\right)
       \\&= (I_N-\alpha(k)\bar{\mathcal{L}})\boldsymbol{\delta}(k)+\alpha(k)(I_N-J)\boldsymbol{w}(k).
    \end{split}
    \end{equation}
    
    By~\eqref{eq:delta} and the definition of $V(k)$, we have
    \begin{equation}
    \begin{split}
        &V(k+1) \\&= \boldsymbol{\delta}^T(k)(I_N-\alpha(k)\bar{\mathcal{L}})^2\boldsymbol{\delta}(k) \\&\quad+2\alpha(k)\boldsymbol{\delta}^T(k)(I_N-\alpha(k)\bar{\mathcal{L}})^T (I_N-J)\boldsymbol{w}(k)\\&\quad+\alpha^2(k)\boldsymbol{w}^T(k)(I_N-J)^2\boldsymbol{w}(k)
        \\&= V(k) -2\alpha(k)\boldsymbol{\delta}^T(k)\bar{\mathcal{L}}\boldsymbol{\delta}(k) +\alpha^2(k)\boldsymbol{\delta}^T(k)\bar{\mathcal{L}}^2\boldsymbol{\delta}(k)\\&\quad+2\alpha(k)\boldsymbol{\delta}^T(k)(I_N-\alpha(k)\bar{\mathcal{L}})(I_N-J)\boldsymbol{w}(k) \\&\quad+\alpha^2(k)\boldsymbol{w}^T(k)(I_N-J)^2\boldsymbol{w}(k) 
        \\&\leq \left(1-2\alpha(k)\lambda_2(\bar{\mathcal{L}})+\alpha^2(k)\left\Vert\bar{\mathcal{L}}\right\Vert^2 \right)V(k)
        \\&\quad+2\alpha(k)\boldsymbol{\delta}^T(k)(I_N-\alpha(k)\bar{\mathcal{L}})(I_N-J) \boldsymbol{w}(k) \\&\quad+\alpha^2(k)\boldsymbol{w}^T(k)(I_N-J)^2\boldsymbol{w}(k).
    \end{split}
    \end{equation}
    \textcolor{black}{
    By Lemma~\ref{lm:property} c) and d), we have
    \begin{equation}
    \begin{split}
        &\mathbb{E}[\boldsymbol{\delta}^T(k)(I_N-\alpha(k)\bar{\mathcal{L}})(I_N-J)\boldsymbol{w}(k)|\mathcal{F}_k]\\
        &=\boldsymbol{\delta}^T(k)(I_N-\alpha(k)\bar{\mathcal{L}})(I_N-J)\mathbb{E}[\Delta\mathcal{L}(k)\boldsymbol{x}(k) + \boldsymbol{v}(k)|\mathcal{F}_k]\\
        &=\boldsymbol{\delta}^T(k)(I_N-\alpha(k)\bar{\mathcal{L}})(I_N-J)
        \\&\quad\times\left(\mathbb{E}[\Delta\mathcal{L}(k)|\mathcal{F}_k]\boldsymbol{x}(k) + \mathbb{E}[\boldsymbol{v}(k)|\mathcal{F}_k]\right)\\
        &= 0 \quad a.s.,
    \end{split}
    \end{equation}
    which implies that
    \begin{equation}\label{eq:1}
    \begin{split}
        \mathbb{E}[\boldsymbol{\delta}^T(k)(I_N-\alpha(k)\bar{\mathcal{L}})(I_N-J)\boldsymbol{w}(k)] = 0.
    \end{split}
    \end{equation}
    By~\eqref{eq:w_square} and~\eqref{eq:E_x_square}, we have $\mathbb{E}\left[\Vert\boldsymbol{w}(k)\Vert^2\right]\leq \overline{M}_1$,
    where 
    \begin{equation}\label{eq:M_1}
    \begin{split}
        \overline{M}_1=& \left(C_L+C_M^\mathrm{(1)}\right) \left(\sup_{l\geq 0}\Phi_1(l:0)\Vert\boldsymbol{x}(0)\Vert^2 \right.\\& + \left.\sup_{k,l\geq 0}\Phi_1(l:k)C_M^\mathrm{(2)}\sum_{t=0}^\infty\alpha^2(t)\right) + C_M^\mathrm{(2)}.
    \end{split}
    \end{equation}
    Hence, we have
    \begin{equation}
    \begin{split}
        \mathbb{E}[\boldsymbol{w}^T(k)(I_N-J)^2\boldsymbol{w}(k)]&\leq \Vert I_N-J\Vert^2 \mathbb{E}\left[\Vert\boldsymbol{w}(k)\Vert^2\right] \\&\leq \Vert I_N-J\Vert^2\overline{M}_1.
    \end{split}
    \end{equation}
    Then, we have
    \begin{equation}\label{eq1}
    \begin{split}
        &\mathbb{E}[V(k+1)]\\&\leq \left(1-2\alpha(k)\lambda_2(\bar{\mathcal{L}}) + \alpha^2(k)\left\Vert\bar{\mathcal{L}}\right\Vert^2\right)\mathbb{E}[V(k)]\\&\quad+ \alpha^2(k)\Vert I_N-J\Vert^2\overline{M}_1
    \end{split}
    \end{equation}
    }
    % The following lemma is needed.
    \begin{lemma}[\cite{polyak1987introduction}]\label{lm:sequence}
        Let $\{u(k)\}_{k\geq 0}$, $\{b(k)\}_{k\geq 0}$, and $\{q(k)\}_{k\geq 0}$ be real sequences. If $0< b(k) \leq 1, q(k) > 0, \forall k\geq 0$, $\sum_{k=0}^\infty b(k) = \infty$, $\lim_{k\rightarrow\infty}\frac{q(k)}{b(k)}=0$, and
        \begin{equation*}
            u(k+1)\leq (1-b(k))u(k)+q(k),
        \end{equation*}
        then $\limsup_{k\rightarrow\infty}u(k)\leq 0$. In particular, if $u(k)\geq 0, \forall k\geq 0$, then $\lim_{k\rightarrow\infty}u(k)=0$.
    \end{lemma}

    By Assumption~\ref{asm:stepsize} we know that
    \begin{equation}
        \sum_{k=0}^\infty \left(2\alpha(k)\lambda_2(\bar{\mathcal{L}}) - \alpha^2(k)\left\Vert(\bar{\mathcal{L}})\right\Vert^2\right) = \infty,
    \end{equation}
    and there exists $\underline{k}\geq 0$ such that $0\leq2\alpha(k)\lambda_2(\bar{\mathcal{L}}) - \alpha^2(k)\left\Vert(\bar{\mathcal{L}})\right\Vert^2<1, \forall k\geq \underline{k}$.
    
    Moreover, we have
    \begin{equation}
        \lim_{k\rightarrow\infty} \frac{\alpha^2(k)\Vert I_N-J\Vert^2\overline{M}_1}{2\alpha(k)\lambda_2(\bar{\mathcal{L}}) - \alpha^2(k)\left\Vert(\bar{\mathcal{L}})\right\Vert^2} = 0.
    \end{equation}
    Hence, by Lemma~\ref{lm:sequence} and \eqref{eq1}, we have $\lim_{m\rightarrow\infty}\mathbb{E}[V(k)]=0$.

\section{Proof of Theorem~\ref{thm:ms}}\label{app:thm_ms}

    Since $\mathbf{1}^T_N\bar{\mathcal{L}} = 0$
    \begin{equation*}
    \begin{split}
         \sum_{i=1}^N x_i(k+1) &= \sum_{i=1}^N x_i(k) + \sum_{i=1}^N \alpha(k)w_i(k)
         \\&= \sum_{i=1}^N x_i(0) + \mathbf{1}^T_N\sum_{t=0}^k \alpha(t)\boldsymbol{w}(t).
    \end{split}
    \end{equation*}
    Since $\sum_{t=0}^k\alpha(t)\boldsymbol{w}(t)$ is adapted to $\mathcal{F}_k$ and 
    \begin{equation}
    \begin{split}
        &\mathbb{E}\left[ \left. \sum_{t=0}^k\alpha(t)\boldsymbol{w}(t) \right| \mathcal{F}_k \right] 
        \\& = \alpha(k)\mathbb{E}\left[ \left. \boldsymbol{w}(k) \right| \mathcal{F}_k\right] 
        + \sum_{t=0}^{k-1}\alpha(t)\boldsymbol{w}(t)
        \\& = \alpha(k)\left(\mathbb{E}[\Delta\mathcal{L}(k)]\boldsymbol{x}(k) + \mathbb{E}[\left. \boldsymbol{v}(k)\right| \mathcal{F}_k]\right) + \sum_{t=0}^{k-1}\alpha(t)\boldsymbol{w}(t)
        \\& = \sum_{t=0}^{k-1}\alpha(t)\boldsymbol{w}(t).
    \end{split}
    \end{equation}

    By $\mathbb{E}\left[\Vert \boldsymbol{w}(k) \Vert^2\right]\leq \overline{M}_1$, we have
    \begin{equation*}
    \begin{split}
       \sup_{k\geq 0} \mathbb{E}\left[ \left\Vert\sum_{t=0}^k\alpha(t)\boldsymbol{w}(t)\right\Vert^2 \right]& \leq  \sup_{k\geq 0}\sum_{t=0}^k\alpha^2(t)\overline{M}_1\\&< \infty,
    \end{split} 
    \end{equation*}
    then by $L^p$ convergence theorem~\cite[Theorem 4.4.6]{durrett2019probability}, we can know that $\sum_{t=0}^k\alpha(k)\boldsymbol{w}(t)$ converges a.s. as $k\rightarrow\infty$. 
    Then $ \frac{1}{N}\sum_{i=1}^N x_i(k)$ converges to $ x^*$ a.s. as $k\rightarrow\infty$, where
    \begin{equation*}
        x^*=\frac{1}{N}\sum_{i=1}^N x_i(0) + \frac{1}{N}\mathbf{1}^T_N\sum_{t=0}^\infty \alpha(t)\boldsymbol{w}(t),
    \end{equation*}
    which satisfies 
    \textcolor{black}{
    \begin{align*}
        \mathbb{E}[x^*]&=\frac{1}{N}\sum_{i=1}^N x_i(0),\\
        Var(x^*)&=\lim_{k\rightarrow\infty} \mathbb{E}\left[\left\Vert\frac{1}{N}\mathbf{1}^T_N\sum_{t=0}^k \alpha(t)\boldsymbol{w}(t)\right\Vert^2\right]
        \\&= \lim_{k\rightarrow\infty} \frac{1}{N^2}\sum_{t=0}^k \mathbb{E}\left[\left(\sum_{i\in\mathcal{V}}\alpha(t)w_i(t)\right)^2\right]
        \\& \leq \lim_{k\rightarrow\infty} \frac{1}{N}\sum_{t=0}^k \sum_{i\in\mathcal{V}}\alpha^2(t)\mathbb{E}\left[ w_i^2(t)\right]
        \\&= \lim_{k\rightarrow\infty} \frac{1}{N}\sum_{t=0}^k \alpha^2(t)\mathbb{E}\left[ \Vert\boldsymbol{w}(t)\Vert^2\right]
        \\&\leq \frac{\overline{M}_1}{N}\sum_{t=0}^\infty \alpha^2(t)
        < \infty.
    \end{align*}
    }
    Therefore, combining with Theorem~\ref{thm:weak}, the mean square average consensus can be achieved.
    
\section{Proof of Theorem~\ref{thm:almost_sure}}\label{app:thm_as}

First, we have
\begin{equation*}
\begin{split}
    \mathbb{E}\left[ V(k+1)| \mathcal{F}_k \right] &\leq \left(1-2\alpha(k)\lambda_2(\bar{\mathcal{L}})+\alpha^2(k)\left\Vert\bar{\mathcal{L}}\right\Vert^2 \right)V(k)\\&\quad + \alpha^2(k)\Vert I_N-J\Vert^2\overline{M}_1.
\end{split}
\end{equation*}
By Siegmund and Robbins Theorem~\cite{robbins1971convergence}, we have $V(k)$ converges almost surely as $k\rightarrow\infty$.
% and 
% \begin{equation*}
%     \sum_{k=0}^\infty \alpha(k)V(k) < \infty \quad a.s.
% \end{equation*}
Then, by Theorem~\ref{thm:weak}, we have
\begin{equation*}
    \lim_{k\rightarrow\infty} V(k) = 0 \quad a.s.
\end{equation*}
Therefore, the system can achieve almost sure consensus.

\section{Proof of Corollary~\ref{coro:weak}}\label{app:coro1}
   Define
    \begin{align*}
        &\bar{\alpha}(k)\triangleq \frac{1}{N}\sum_{i\in\mathcal{V}}\alpha_i(k),\\
        \Delta &A(k)\triangleq\\&\mathrm{diag}\left(\bar{\alpha}(k)-\alpha_1(k), \bar{\alpha}(k)-\alpha_2(k), \ldots, \bar{\alpha}(k)-\alpha_N(k)\right).
    \end{align*}
    We have
    \begin{equation}
        \begin{split}
        &V(k+1) \\&= \boldsymbol{\delta}^T(k)(I_N-\bar{\alpha}(k)\bar{\mathcal{L}})^2\boldsymbol{\delta}(k) 
        \\&\quad+\boldsymbol{\delta}^T(k)\bar{\mathcal{L}}\Delta A(k)(I_N-J)\Delta A(k)\bar{\mathcal{L}}\boldsymbol{\delta}(k) 
        \\&\quad+2\boldsymbol{\delta}^T(k)(I_N-\bar{\alpha}(k)\bar{\mathcal{L}})(I_N-J)\Delta A(k)\bar{\mathcal{L}}\boldsymbol{\delta}(k) 
        \\&\quad+2\boldsymbol{\delta}^T(k)(I_N-\alpha(k)\bar{\mathcal{L}})^T (I_N-J)(\bar{\alpha}(k)I_N-\Delta A(k))\boldsymbol{w}(k)
        \\&\quad+2\boldsymbol{\delta}^T(k)\bar{\mathcal{L}}\Delta A(k)(I_N-J)(\bar{\alpha}(k)I_N-\Delta A(k))\boldsymbol{w}(k)
        \\&\quad+\boldsymbol{w}^T(k)(\bar{\alpha}(k)I_N-\Delta A(k))(I_N-J)(\bar{\alpha}(k)I_N-\Delta A(k))\\&\quad\times\boldsymbol{w}(k)
        \\&\leq (1-s(t))V(k) +\left\Vert \bar{\alpha}(k)I_N-\Delta A(k)\right\Vert^2\Vert I_N-J\Vert^2\overline{M}_1,
    \end{split}
    \end{equation}
    where
    \begin{equation*}
        \begin{split}
            s(k) =& 2\lambda_2(\bar{\mathcal{L}})\bar{\alpha}(k)-\bar{\alpha}^2(k)\Vert \bar{\mathcal{L}} \Vert^2 - \left\Vert\Delta A(k)\right\Vert^2\left\Vert\bar{\mathcal{L}}\right\Vert^2
            \\&- 2\left(1+\max_{i\in\mathcal{V}}\sup_{t\geq 0}\alpha_i(t)\left\Vert \bar{\mathcal{L}}\right\Vert\right)\left\Vert\Delta A(k)\right\Vert \left\Vert\bar{\mathcal{L}}\right\Vert.
        \end{split}
    \end{equation*}
    By $\sum_{k=0}^\infty \alpha_i(k)=\infty$ and $ \sum_{k=0}^\infty \alpha_i^2(k)<\infty, \forall i\in\mathcal{V}$, we know that $\sum_{k=0}^\infty \bar{\alpha}(k)=\infty$ and $ \sum_{k=0}^\infty \bar{\alpha}^2(k)<\infty, \forall i\in\mathcal{V}$. By $\max_{i,j\in\mathcal{V}}\left|\alpha_i(k)-\alpha_j(k)\right| = o\left(\sum_{i\in\mathcal{V}}\alpha_i(k)\right), k\rightarrow\infty$, we know that $\left\Vert\Delta A(k)\right\Vert = o\left(\bar{\alpha}(k)\right), k\rightarrow\infty$. Then, we know there exists $\underline{k}_1 > 0$ such that $0<s(k)\leq 1, \forall k\geq\underline{k}_1$, $\sum_{k=\underline{k}_1}^\infty s(k)=\infty$, and $\frac{\bar{\alpha}^2(k)}{s(k)}\rightarrow 0, k\rightarrow 0$. Then similar to Theorem~\ref{thm:weak}, we have~\eqref{eq:weak2}.

\section{Proof of Corollary~\ref{coro:ms}}\label{app:coro2}

According to~\eqref{eq:protocol2}, we have
\begin{equation*}
    \begin{split}
        \boldsymbol{x}(k+1) =& \left(I_N - A(k)\bar{\mathcal{L}} \right) \boldsymbol{x}(k) + A(k)\boldsymbol{w}(k)\\
        =&\prod_{t=0}^k \left(I_N - A(t)\bar{\mathcal{L}} \right) \boldsymbol{x}(0) \\&+ \sum_{t=0}^k \prod_{i=t+1}^k \left(I_N - A(i)\bar{\mathcal{L}} \right)A(t)\boldsymbol{w}(t).
    \end{split}
\end{equation*}
Then, we have
\begin{equation*}
    \begin{split}
        \sum_{i=1}^N x_i(k+1)=&\mathbf{1}^T\prod_{t=0}^k \left(I_N - A(t)\bar{\mathcal{L}} \right) \boldsymbol{x}(0) \\&+ \mathbf{1}^T\sum_{t=0}^k \prod_{i=t+1}^k \left(I_N - A(i)\bar{\mathcal{L}} \right)A(t)\boldsymbol{w}(t).
    \end{split}
\end{equation*}
% \begin{equation*}
%     \begin{split}
%          \sum_{i=1}^N x_i(k+1) &= \sum_{i=1}^N x_i(k) + \mathbf{1}^T\Delta A(k) \bar{\mathcal{L}}\boldsymbol{x}(k) + \mathbf{1}^T A(k)\boldsymbol{w}(k)
%          \\&= \sum_{i=1}^N x_i(0) + \sum_{t=0}^k \mathbf{1}^T\Delta A(t) \bar{\mathcal{L}}\boldsymbol{x}(k) + \sum_{t=0}^k\mathbf{1}^T A(t)\boldsymbol{w}(t).
%     \end{split}
% \end{equation*}
    Since $\sum_{t=0}^k \prod_{i=t+1}^k \left(I_N - A(i)\bar{\mathcal{L}} \right)A(t)\boldsymbol{w}(t)$ is adapted to $\mathcal{F}_k$ and 
    \begin{equation*}
    \begin{split}
        &\mathbb{E}\left[ \left. \sum_{t=0}^k \prod_{i=t+1}^k \left(I_N - A(i)\bar{\mathcal{L}} \right)A(t)\boldsymbol{w}(t) \right| \mathcal{F}_k \right] 
        \\& = A(k)\mathbb{E}\left[ \left. \boldsymbol{w}(k) \right| \mathcal{F}_k\right] 
        + \sum_{t=0}^{k-1} \prod_{i=t+1}^{k-1} \left(I_N - A(i)\bar{\mathcal{L}} \right)A(t)\boldsymbol{w}(t)
        \\& = \sum_{t=0}^{k-1} \prod_{i=t+1}^{k-1} \left(I_N - A(i)\bar{\mathcal{L}} \right)A(t)\boldsymbol{w}(t),
    \end{split}
    \end{equation*}
    By $\mathbb{E}\left[\Vert \boldsymbol{w}(k) \Vert^2\right]\leq \overline{M}_1$, we have
    \begin{equation*}
    \begin{split}
       &\sup_{k\geq 0} \mathbb{E}\left[ \left\Vert\sum_{t=0}^k \prod_{i=t+1}^k \left(I_N - A(i)\bar{\mathcal{L}} \right)A(t)\boldsymbol{w}(t)\right\Vert^2 \right]
       \\&= \sup_{k\geq 0} \sum_{t=0}^k\mathbb{E}\left[ \left\Vert \prod_{i=t+1}^k \left(I_N - A(i)\bar{\mathcal{L}} \right)A(t)\boldsymbol{w}(t)\right\Vert^2 \right]
       \\&\leq \sup_{k\geq 0} \left(\sup_{j\geq 0}\left\Vert\prod_{i=j}^k \left(I_N - A(i)\bar{\mathcal{L}} \right)\right\Vert^2 \sum_{t=0}^k\mathbb{E}\left[ \left\Vert A(t)\boldsymbol{w}(t)\right\Vert^2 \right]\right)
       \\&\leq  \overline{M}_1\max_{i\in\mathcal{V}}\sup_{k\geq 0} \left(\sup_{j\geq 0}\left\Vert\prod_{i=j}^k \left(I_N - A(i)\bar{\mathcal{L}} \right)\right\Vert^2 \sum_{t=0}^k\alpha_i^2(t)\right)\\&< \infty,
    \end{split} 
    \end{equation*}
    then by $L^p$ convergence theorem~\cite[Theorem 4.4.6]{durrett2019probability}, we can know that $\sum_{t=0}^k \prod_{i=t+1}^k \left(I_N - A(i)\bar{\mathcal{L}} \right)A(t)\boldsymbol{w}(t)$ converges a.s. as $k\rightarrow\infty$.\\   
    Moreover, since the product of stochastic matrices $I_N - A(k)\bar{\mathcal{L}}, k=0,1,\ldots$, is still a stochastic matrix, the matrix $\prod_{k=0}^\infty (I_N - A(k)\bar{\mathcal{L}})$ is also a stochastic matrix. Therefore, all the elements of the vector $\prod_{k=0}^\infty (I_N - A(k)\bar{\mathcal{L}}) \boldsymbol{x}(0)$ fall into the convex hull formed by $x_i(0), i = 1, 2, \ldots, N$, i.e., the interval $\left[\min_{i\in\mathcal{V}} x_i(0), \max_{i\in\mathcal{V}} x_i(0)\right]$, and hence $\frac{\mathbf{1}^T}{N}\prod_{t=0}^\infty \left(I_N - A(t)\bar{\mathcal{L}} \right) \boldsymbol{x}(0)$ is bounded.\\
    Combining with Corollary~\ref{coro:weak}, $\boldsymbol{x}(k)$ converges to $\mathbf{1}_N x^*$ in mean square, where
    \begin{equation*}
    \begin{split}
        x^*=&\frac{\mathbf{1}^T}{N}\prod_{t=0}^\infty \left(I_N - A(t)\bar{\mathcal{L}} \right) \boldsymbol{x}(0) \\&+ \frac{\mathbf{1}^T}{N}\sum_{t=0}^\infty \prod_{i=t+1}^\infty \left(I_N - A(i)\bar{\mathcal{L}} \right)A(t)\boldsymbol{w}(t).
    \end{split}
    \end{equation*}
    which satisfies 
    \begin{align*}
        &\mathbb{E}[x^*]=\frac{\mathbf{1}^T}{N}\prod_{t=0}^\infty \left(I_N - A(t)\bar{\mathcal{L}} \right) \boldsymbol{x}(0),\\
        &Var(x^*)\\&=\lim_{k\rightarrow\infty} \mathbb{E}\left[\left\Vert\frac{\mathbf{1}^T}{N}\sum_{t=0}^k\prod_{i=t+1}^k \left(I_N - A(i)\bar{\mathcal{L}} \right)A(t)\boldsymbol{w}(t)\right\Vert^2\right]
        \\&= \lim_{k\rightarrow\infty} \frac{1}{N^2}\sum_{t=0}^k \mathbb{E}\left[\left\Vert \mathbf{1}^T\prod_{i=t+1}^k \left(I_N - A(i)\bar{\mathcal{L}} \right)A(t)\boldsymbol{w}(t)\right\Vert^2\right]
        \\& \leq \lim_{k\rightarrow\infty} \frac{1}{N}\sum_{t=0}^k \mathbb{E}\left[\left\Vert \prod_{i=t+1}^k \left(I_N - A(i)\bar{\mathcal{L}} \right)A(t)\boldsymbol{w}(t)\right\Vert^2\right] \\&\leq\frac{\overline{M}_1}{N}\sup_{k,t\geq 0}\left\Vert\prod_{t}^k \left(I_N - A(t)\bar{\mathcal{L}}\right) \right\Vert^2 \max_{i\in\mathcal{V}}\sum_{t=0}^\infty\alpha_i^2(t)\\&< \infty.
    \end{align*}
    Therefore, the mean square consensus can be achieved.

\section{Proof of Corollary~\ref{coro:as}}\label{app:coro3}
First, we have
    \begin{equation*}
    \begin{split}
        \mathbb{E}\left[ V(k+1)| \mathcal{F}_k \right] &\leq (1-s(t)+2\lambda_2(\bar{\mathcal{L}})\bar{\alpha}(k))V(k) \\&\quad + \left\Vert \bar{\alpha}(k)I_N-\Delta A(k)\right\Vert^2\Vert I_N-J\Vert^2\overline{M}_1\\ &\quad -2\lambda_2(\bar{\mathcal{L}})\bar{\alpha}(k)V(k).
    \end{split}
    \end{equation*}
    Moreover, the condition~c) in Corollary~\ref{coro:as} implies the condition~c) in Corollary~\ref{coro:weak}. 
    Then similar to Theorem~\ref{thm:almost_sure}, by Siegmund and Robbins Theorem~\cite{robbins1971convergence}, we have $V(k)$ converges almost surely as $k\rightarrow\infty$.
    % and 
    % \begin{equation*}
    %     \sum_{k=0}^\infty \alpha(k)V(k) < \infty \quad a.s.
    % \end{equation*}
    Then, by Corollary~\ref{coro:weak}, we have
    \begin{equation*}
        \lim_{k\rightarrow\infty} V(k) = 0 \quad a.s.
    \end{equation*}
    Therefore, the system can achieve almost sure consensus.

\section{Proof of Theorem~\ref{thm:4.1}}\label{app:thm4.1}
    % By the definition of $\delta(k)$, we have 
    By~\eqref{eq:delta}, we have
    \begin{equation}
        \begin{split}
        &\delta((m+1)L)\\&=P((m+1)L:mL)\delta(mL)+W((m+1)L:mL),
        \end{split}
    \end{equation}
    where $P(l:k)=\prod_{i=l}^{k}(I_N-\alpha(i)\bar{\mathcal{L}}(i))$ and $W(l:k) = \sum_{i=k}^l P(l:i+1)\alpha(i)(I_N-J)\boldsymbol{w}(i)$.
    
    According to Assumption~\ref{asm:stepsize2},  there exists a constant $C$ and a positive integer $m_0$ such that $\alpha(mL)\leq C\alpha((m+1)L)$ and $\alpha(mL)\leq 1, \forall m\geq m_0$. Then we have 
    \begin{equation}\label{eq:P}
    \begin{split}
        &\left\Vert P^T((m+1)L:mL)P((m+1)L:mL)\right.\\&
        \left. - I_N + \sum_{i=mL}^{(m+1)L-1}\alpha(i)\left(\bar{\mathcal{L}}(i) + \bar{\mathcal{L}}^T(i)\right) \right\Vert
        \\&=\left\Vert P^T((m+1)L:mL)P((m+1)L:mL) \right.\\&\quad
        \left. - I_N+ 2\sum_{i=mL}^{(m+1)L-1}\alpha(i)\bar{\mathcal{L}}(i) \right\Vert
        \\&\leq \alpha^2(mL)(2^{2L}-2L-1)(\max\{\sup_{k\geq 0}\Vert \bar{\mathcal{L}}(k)\Vert,1\})^{2L}
        \\&\leq\alpha^2((m+1)L)M_L,
    \end{split}
    \end{equation}
    where $M_L = C^2(2^{2L}-2L-1)(\max\{\sup_{k\geq 0}\Vert \bar{\mathcal{L}}(k)\Vert,1\})^{2L}$.

By \eqref{eq:v_square}~-~\eqref{eq:w_square}, we have
\begin{equation}\label{eq:E_x_square_2}
\begin{split}
     &\mathbb{E}\left[\Vert \boldsymbol{x}(k+1)\Vert^2\right] \\&= \mathbb{E}\left[ \Vert \left(I_N - \alpha(k)\bar{\mathcal{L}}(k) \right) \boldsymbol{x}(k) + \alpha(k)\boldsymbol{w}(k)\Vert^2\right]
     \\& = \mathbb{E}\left[ \Vert \left(I_N - \alpha(k)\bar{\mathcal{L}}(k) \right) \boldsymbol{x}(k) \Vert^2 \right] + \alpha^2(k)\mathbb{E}\left[ \Vert \boldsymbol{w}(k)\Vert^2\right]
     \\&\leq \Vert I_N - \alpha(k)\bar{\mathcal{L}}(k)  \Vert^2 \mathbb{E}\left[ \Vert\boldsymbol{x}(k) \Vert^2\right]  + \alpha^2(k)\mathbb{E}\left[ \Vert \boldsymbol{w}(k)\Vert^2\right]
     \\&\leq \left(\Vert I_N - \alpha(k)\bar{\mathcal{L}}(k)\Vert^2+\alpha^2(k)\left(C_L+C_M^\mathrm{(1)}\right) \right) \mathbb{E}\left[ \Vert\boldsymbol{x}(k) \Vert^2\right] \\&\quad + \alpha^2(k)C_M^\mathrm{(2)}
     % \\&=\prod_{i=k}^0\left(\Vert I_N - \alpha(i)\bar{\mathcal{L}}(i)\Vert^2+\alpha^2(i)\left(C_L+C_M^\mathrm{(1)}\right) \right) \mathbb{E}\left[ \Vert\boldsymbol{x}(0) \Vert^2\right] \\&\quad + \sum_{i=0}^k\alpha^2(i)C_M^\mathrm{(2)}\prod_{j=k}^i\left(\Vert I_N - \alpha(j)\bar{\mathcal{L}}(j)\Vert^2+\alpha^2(j)\left(C_L+C_M^\mathrm{(1)}\right) \right)
     % \\&\leq \left(1-2\alpha(k)\lambda_2\left(\sum_{i=0}^{k}\bar{\mathcal{L}}(i)\right) + \alpha^2(0)(2^{2L}-2L-1)(\max\{\sup_{k\geq 0}\Vert \bar{\mathcal{L}}(k)\Vert,1\})^{2L} \right)
\end{split}
\end{equation}
By~\eqref{eq:P} and~\eqref{eq:E_x_square_2}, we have
\begin{equation}
\begin{split}
        &\mathbb{E}\left[\Vert \boldsymbol{x}((m+1)L)\Vert^2\right] \\&\leq \prod_{i=mL}^{(m+1)L-1}\left(\Vert I_N - \alpha(i)\bar{\mathcal{L}}(i)\Vert^2 + \alpha^2(i)\left(C_L+C_M^\mathrm{(1)}\right) \right) \\&\quad \times\mathbb{E}\left[ \Vert\boldsymbol{x}(mL) \Vert^2\right] + \sum_{i=mL}^{(m+1)L-1}\alpha^2(i)C_M^\mathrm{(2)} \Phi_3((m+1)L-1:i)
        \\&\leq \left(1-2\lambda_2\left(\sum_{i=mL}^{(m+1)L-1}\bar{\mathcal{L}}(i)\right)\alpha((m+1)L) \right.\\&\left. \quad +\alpha^2((m+1)L)B_1 \right)\mathbb{E}\left[\Vert\boldsymbol{x}(mL)\Vert^2\right] \\&\quad+ \sum_{i=mL}^{(m+1)L-1}\alpha^2(i)C_M^\mathrm{(2)} \Phi_3((m+1)L-1:i)
        \\&\leq \Phi_2(m:0)\Vert\boldsymbol{x}(0)\Vert^2 \\&\quad +\sum_{l = 0}^m \Phi_2(m:l)\sum_{i=lL}^{(l+1)L-1}\alpha^2(i)C_M^\mathrm{(2)} \Phi_3((m+1)L-1:i),
\end{split}
\end{equation}
where 
\begin{align*}
    &B_1\triangleq M_L + C^2\left(C_L+C_M^\mathrm{(1)}\right),\\
    &\Phi_2(k:l)\triangleq  \prod_{j=l}^k \left(1-2\inf_{j\geq 0}\lambda_2\left(\sum_{i=jL}^{(j+1)L-1}\bar{\mathcal{L}}(i)\right)\alpha((j+1)L)\right.\\&\qquad\qquad\left.+\alpha^2((j+1)L)B_1 \right),\\
    &\Phi_3(k:l) \triangleq \prod_{j=k}^{l}\left(\Vert I_N - \alpha(j)\bar{\mathcal{L}}(j)\Vert^2+\alpha^2(j)\left(C_L+C_M^\mathrm{(1)}\right) \right).
\end{align*}
Moreover, for any $k_0 \in \{0, 1, 2, \ldots, L-1\}$, we have
\begin{equation}
\begin{split}
    &\mathbb{E}\left[\Vert \boldsymbol{x}(k_0+mL)\Vert^2\right] \\
    &\leq \prod_{i=mL}^{k_0+mL-1}\left(\Vert I_N - \alpha(i)\bar{\mathcal{L}}(i)\Vert^2 +\alpha^2(i)\left(C_L+C_M^\mathrm{(1)}\right) \right) \\&\quad \times\mathbb{E}\left[ \Vert\boldsymbol{x}(mL) \Vert^2\right] + \sum_{i=mL}^{k_0+mL-1}\alpha^2(i)C_M^\mathrm{(2)} \Phi_3(k_0+mL-1:i)
    \\&\leq \prod_{i=mL}^{k_0+mL-1}\left(\Vert I_N - \alpha(i)\bar{\mathcal{L}}(i)\Vert^2 +\alpha^2(i)\left(C_L+C_M^\mathrm{(1)}\right) \right) \\&\quad \times\left(\Phi_2(m-1:0)\Vert\boldsymbol{x}(0)\Vert^2 \right.\\&\quad \left.+\sum_{l = 0}^{m-1} \Phi_2(m-1:l)\sum_{i=lL}^{(l+1)L-1}\alpha^2(i)C_M^\mathrm{(2)} \Phi_3(mL-1:i)\right) \\& \quad+ \sum_{i=mL}^{k_0+mL-1}\alpha^2(i)C_M^\mathrm{(2)} \Phi_3(k_0+mL-1:i).
\end{split}
\end{equation}
Then we have $\mathbb{E}\left[\Vert \boldsymbol{w}(k_0+mL) \Vert^2\right]\leq \overline{M}_2$, where
\begin{equation}\label{eq:M_2}
\begin{split}
    \overline{M}_2&=\left(C_L+C_M^\mathrm{(1)}\right) \sup_{m\geq 0}\left(\sup_{k_0\geq 0}\prod_{i=mL}^{k_0+mL-1}\left(\Vert I_N - \alpha(i)\bar{\mathcal{L}}(i)\Vert^2 \right. \right.\\&\quad\left.\left.+\alpha^2(i)\left(C_L+C_M^\mathrm{(1)}\right) \right)\left(\Phi_2(m-1:0)\Vert\boldsymbol{x}(0)\Vert^2 \right.\right.\\&\quad \left.\left.+\sum_{l = 0}^{m-1} \Phi_2(m-1:l)\sum_{i=lL}^{(l+1)L-1}\alpha^2(i)C_M^\mathrm{(2)} \Phi_3(mL-1:i)\right)\right. \\& \quad \left.+ \sup_{k_0\geq 0}\sum_{i=mL}^{k_0+mL-1}\alpha^2(i)C_M^\mathrm{(2)} \Phi_3(k_0+mL-1:i)\right) \\&\quad + C_M^\mathrm{(2)}.
\end{split}
\end{equation}
    
    By~\eqref{eq:delta} and the definition of $V(k)$, we have
    \begin{equation}\label{ieq:2}
    \begin{split}
        &V((m+1)L) \\&= \delta^T(mL)P^T((m+1)L:mL)P((m+1)L:mL)\delta(mL)\\&\quad+2\delta^T(mL)P^T((m+1)L:mL)W((m+1)L:mL)\\&\quad+W((m+1)L:mL)^T W((m+1)L:mL)
        \\&= \delta^T(mL)\left(P^T((m+1)L:mL)P((m+1)L:mL)\right.\\&\quad\left.-I_N + 2\sum_{i=mL}^{(m+1)L-1}\alpha(i)\bar{\mathcal{L}}(i)\right)\delta(mL) + V(mL)\\&\quad - 2\delta^T(mL)\left(\sum_{i=mL}^{(m+1)L-1}\alpha(i)\bar{\mathcal{L}}(i)\right)\delta(mL)\\&\quad+2\delta^T(mL)P^T((m+1)L:mL)W((m+1)L:mL)\\&\quad+W((m+1)L:mL)^T W((m+1)L:mL)
        \\&\leq V(mL) + \alpha^2((m+1)L)M_L V(mL)\\&\quad - 2\delta^T(mL)\left(\sum_{i=mL}^{(m+1)L-1}\alpha(i)\bar{\mathcal{L}}(i)\right)\delta(mL) 
    \\&\quad+2\delta^T(mL)P^T((m+1)L:mL) W((m+1)L:mL)\\&\quad+W((m+1)L:mL)^T W((m+1)L:mL)
        \\&\leq V(mL) + \alpha^2((m+1)L)M_L V(mL)\\&\quad - 2\alpha((m+1)L)\delta^T(mL)\left(\sum_{i=mL}^{(m+1)L-1}\bar{\mathcal{L}}(i)\right)\delta(mL) \\&\quad+2\delta^T(mL)P^T((m+1)L:mL)W((m+1)L:mL)\\&\quad+W((m+1)L:mL)^T W((m+1)L:mL).
    \end{split}
    \end{equation}
    % where $\hat{\mathcal{L}}(i)=\frac{\bar{\mathcal{L}}(i)+\bar{\mathcal{L}}^T(i)}{2}$.

    Similar to~\eqref{eq:1}, we have
    \begin{equation}
        \mathbb{E}[\delta^T(mL)P^T((m+1)L:mL)W((m+1)L:mL)] = 0
    \end{equation}
   %  \textcolor{black}{
   %  Moreover, by~\eqref{eq:E_x_square} and ~\eqref{eq:w_square} we have $\mathbb{E}\left[ \Vert \boldsymbol{w}(k)\Vert^2 \right] \leq \overline{M}_2$, where
   %  \begin{equation*}
   %  \begin{split}
   %       \overline{M}_2&\triangleq \left(C_L+C_M^\mathrm{(1)}\right) \left(\sup_{k\geq 0}\Phi_2(k:0)\Vert\boldsymbol{x}(0)\Vert^2 \right.\\& \quad+ \left.\sum_{t=0}^k\sup_{k\geq 0}\Phi_2(k:t+1)\alpha^2(t)C_M^\mathrm{(2)}\right) + C_M^\mathrm{(2)}.
   %  \end{split}
   %  \end{equation*}
   %  and
   % \begin{equation*}
   % \begin{split}
   %    &\Phi_2(k:l)\\&\triangleq \prod_{t=k}^l\left(1-2\alpha(t)\lambda_2(\bar{\mathcal{L}}(k))+\alpha^2(t)\left(\lambda_2^2(\bar{\mathcal{L}}(k))+C_L+C_M^\mathrm{(1)}\right) \right).
   % \end{split}
   % \end{equation*}
   % }
    % Hence, we have
    % \begin{equation}
    % \begin{split}
    %     &\mathbb{E}\left[ (P(l:i+1)\alpha(i)\boldsymbol{w}(i))^T (P(l:j+1)\alpha(j)\boldsymbol{w}(j))\right]
    %     \\&\leq  \max_{\xi\in\{i,j\}}\alpha^2(\xi)\mathbb{E}\left[ \Vert P(l:\xi+1)\boldsymbol{w}(\xi)\Vert^2 \right]
    %     \\&\leq \max_{\xi\in\{i,j\}}\alpha^2(\xi)\mathbb{E}\left[ \Vert P(l:\xi+1)\Vert^2 \right] \mathbb{E}\left[ \Vert \boldsymbol{w}(\xi)\Vert^2 \right]
    %     \\&\leq \overline{M}_2 (\max_{k\in\{i,j\}}\alpha^2(k))2^{2L}(\max\{\sup_{k\geq 0}\Vert \bar{\mathcal{L}}(k)\Vert,1\})^{2L}.
    % \end{split}   
    % \end{equation}
    
    By $\mathbb{E}\left[\Vert \boldsymbol{w}(k_0+mL) \Vert^2\right]\leq \overline{M}_2$, we have
    \begin{equation}
    \begin{split}
         &\mathbb{E}\left[\Vert W((m+1)L:mL)\Vert^2\right]\\
         &= \mathbb{E}\left[\left\Vert\sum_{i=mL}^{(m+1)L-1} P((m+1)L-1:i)\alpha(i)(I_N-J)\boldsymbol{w}(i)\right\Vert^2\right]\\
         &=\sum_{i=mL}^{(m+1)L-1}  \mathbb{E}\left[\left\Vert P((m+1)L-1:i)\alpha(i)(I_N-J)\boldsymbol{w}(i)\right\Vert^2\right]\\
         &\leq \sum_{i=mL}^{(m+1)L-1}  \alpha^2(i)\Vert I_N-J\Vert^2 \Vert P((m+1)L-1:i)\Vert^2\\&\quad\times\mathbb{E}\left[\Vert\boldsymbol{w}(i)\Vert^2\right]\\
         &= \sum_{i=mL}^{(m+1)L-1}  \alpha^2(i)\Vert I_N-J\Vert^2 \Vert P((m+1)L-1:i)\Vert^2\\&\quad\times\mathbb{E}\left[\Vert\boldsymbol{w}(i)\Vert^2\right]\\
         &\leq N_L \sum_{i=mL}^{(m+1)L-1}\alpha^2(i),
    \end{split}
    \end{equation}
    where $N_L \triangleq \overline{M}_2 \Vert I_N-J\Vert^2 \sup_{m\geq 0}\Vert P((m+1)L-1:i)\Vert^2$.

    Then for all $m\geq m_0$, we have
    \begin{equation}\label{ieq:1}
    \begin{split}
        &\mathbb{E}[V((m+1)L)]\\&\leq \left(1-2\lambda_2\left(\sum_{i=mL}^{(m+1)L-1}\bar{\mathcal{L}}(i)\right)\alpha((m+1)L)\right.\\&\quad\left.+\alpha^2((m+1)L)M_L\right)\mathbb{E}[V(mL)]+N_L \sum_{i=mL}^{(m+1)L-1}\alpha^2(i)
        \\&\leq \left(1-2\inf_{m\geq 0}\lambda_2\left(\sum_{i=mL}^{(m+1)L-1}\bar{\mathcal{L}}(i)\right)\alpha((m+1)L)\right.\\&\quad\left.+\alpha^2((m+1)L)M_L\right)\mathbb{E}[V(mL)]+N_L \sum_{i=mL}^{(m+1)L-1}\alpha^2(i),
    \end{split}
    \end{equation}

    By Lemma~\ref{lm:sequence} and \eqref{ieq:1}, we have $\lim_{m\rightarrow\infty}\mathbb{E}[V(mL)]=0$.
    
    Then, for any given $\epsilon>0$, there exists an $m_1>0$ such that $\mathbb{E}[V(mL)]\leq \epsilon, \forall m\geq m_1$ and $\alpha^2(t)<\epsilon, \forall t\geq m_1 L$.
    Let $m_k = \lfloor k/L \rfloor$, then for any $k\geq m_1 L$ we have $m_k\geq m_1$ and
    \begin{equation*}
        0\leq k - m_k L < L.
    \end{equation*}
    
    From the definition of $V(k)$, we have
    \begin{equation}
    \begin{split}
        \mathbb{E}[V(k)]\leq& \phi(k:m_k L)\mathbb{E}[V(m_k L)] \\& + \beta\sum_{i=m_k L}^{k-1}\phi(k-1:i)\alpha^2(i), \forall k\geq m_1 L,
    \end{split}
    \end{equation}
    where $\beta=\Vert I_N-J\Vert^2\overline{M}_1$, $\phi(k:l) = \prod_{i=l}^{k-1}(1-2\lambda_2(\bar{\mathcal{L}}(i))\alpha(i)+\alpha^2(i)\Vert\bar{\mathcal{L}}(i)\Vert^2), \forall k>0, 0\leq i<k$, and $\phi(i,i)=1, \forall i\geq 0$. There exists a constant $\gamma\geq 1$ such that $\phi(k:l)\leq \gamma^{k-l}, \forall k\geq l\geq 0$. Therefore, we have
    \begin{equation}
        % \begin{split}
        \mathbb{E}[V(k)]\leq \gamma^L\epsilon + \gamma^L\beta\sum_{i=m_k L}^{k-1}\alpha^2(i)
        \leq \gamma^L(1+\beta L)\epsilon, \forall k\geq m_1 L.
    % \end{split}
    \end{equation}
    Then \eqref{eq:lyapunov2} holds due to the arbitrariness of $\epsilon$.

\section{Proof of Theorem~\ref{thm:4.2}}\label{app:thm4.2}

    Since $\mathbf{1}^T_N\mathcal{L}(k) = 0$
    \begin{equation*}
    \begin{split}
         \sum_{i=1}^N x_i(k+1) &= \sum_{i=1}^N x_i(k) + \alpha(k)\sum_{i=1}^N w_i(k)
         \\&= \sum_{i=1}^N x_i(0) + \sum_{t=0}^k\alpha(t)\sum_{i=1}^N w_i(t).
    \end{split}
    \end{equation*}
    By Theorem 4.4.6 in~\cite{durrett2019probability}, we can know that $\sum_{t=0}^k\alpha(t)\sum_{i=1}^N w_i(t)$ converges in mean square as $k\rightarrow\infty$. Then we have
    \begin{equation*}
    \begin{split}
          x^*=\frac{1}{N}\sum_{i=1}^N x_i(0) + \frac{1}{N}\sum_{t=0}^\infty\alpha(t)\sum_{i=1}^N w_i(t),
    \end{split}
    \end{equation*}
    and 
    \textcolor{black}{
    \begin{align*}
        \mathbb{E}[x^*]&=\frac{1}{N}\sum_{i=1}^N x_i(0),\\
        Var(x^*) &=\mathbb{E}\left[\left(\frac{1}{N}\sum_{t=0}^\infty\alpha(t)\sum_{i=1}^N w_i(t)\right)^2\right]
        \\&\leq \frac{\overline{M}_2}{N}\sum_{t=0}^\infty \alpha^2(t)
        < \infty.
    \end{align*}
    }

\section{Proof of Lemma~\ref{lm:as}}\label{app:lm_as}

    By Theorem~\ref{thm:weak}, we have 
    \begin{equation}\label{eq:weak}
        \lim_{k\rightarrow\infty} \mathbb{E}[V(k)] = 0.
    \end{equation}
    
    By $\mathbb{E}\left[\Vert \boldsymbol{w}(k_0+mL) \Vert^2\right]\leq \overline{M}_2$, there exists a constant $\tilde{N}_L >0$ such that
    \begin{equation}\label{ieq:3}
    \begin{split}
        &\sup_{m\geq 0}\mathbb{E}\left[\Vert W(k_0+(m+1)L_0:k_0+m L_0)\Vert^2|\mathcal{F}_{k_0+m L_0} \right]\\&\leq\tilde{N}_L \sup_{k\geq 0, m>0}\mathbb{E}\left[\Vert \boldsymbol{w}(k+m)\Vert^2|\mathcal{F}_{k} \right] %\\&\quad\times
        \sum_{i=k_0+m L_0}^{k_0+(m+1)L_0-1}\alpha^2(i)
    \end{split}
    \end{equation}
     
     Similar to~\eqref{ieq:2}, we have
     \begin{equation}
     \begin{split}
         &V(k_0+(m+1)L_0)
        \\&\leq \left(1-2\inf_{m\geq 0}\lambda_2\left(\sum_{i=k_0+mL_0}^{k_0+(m+1)L_0-1}\bar{\mathcal{L}}(i)\right)\alpha(k_0+(m+1)L_0)\right.\\&\quad\left.+\alpha^2(k_0+(m+1)L_0)M_L\right)V(k_0+mL)\\&\quad - 2\alpha(k_0+(m+1)L_0)\delta^T(k_0+mL_0)\\&\quad\times\left(\sum_{i=k_0+mL_0}^{k_0+(m+1)L_0-1}\bar{\mathcal{L}}(i)\right)\delta(k_0+mL_0) \\&\quad+2\delta^T(k_0+mL_0)P^T(k_0+(m+1)L_0:k_0+mL_0)\\&\quad\times W(k_0+(m+1)L_0:k_0+mL_0)\\&\quad+\Vert W(k_0+(m+1)L_0:k_0+mL_0) \Vert^2.
     \end{split}
     \end{equation} 
     Since $V(k_0+m L_0)$ is adapted to $\mathcal{F}_{k_0+m L_0}$, we have
     \begin{equation}\label{ieq:4}
     \begin{split}
         &\mathbb{E}\left[V(k_0+(m+1)L_0)|\mathcal{F}_{k_0+m L}\right]\\&\leq \left(1+\alpha^2(k_0+(m+1)L_0)M_L\right)V(k_0+mL_0) \\&\quad +\tilde{N}_L \sup_{k\geq 0, m>0}\mathbb{E}\left[\Vert \boldsymbol{w}(k+m)\Vert^2|\mathcal{F}_{k} \right] %\\&\quad\times
        \sum_{i=k_0+m L_0}^{k_0+(m+1)L_0-1}\alpha^2(i).
     \end{split}
     \end{equation}
    By Theorem 4.2.12 in~\cite{durrett2019probability} and \eqref{eq:weak}, we have
    \begin{equation*}
         \lim_{m\rightarrow\infty} V(k_0+mL_0) = 0 \quad a.s.
    \end{equation*}

\section{Proof of Theorem~\ref{thm:4.3}}\label{app:thm4.3}

By Assumption~\ref{asm:joint_connectivity}, we have $\inf_{m\geq 0}\lambda_2\left(\sum_{i=l+mL}^{l+(m+1)L-1}\mathcal{L}(i)\right)>0, \forall l\in\{0,1,\ldots, L-1\}$. Then by Lemma~\ref{lm:as}, we have
    \begin{equation*}
        \lim_{m\rightarrow\infty}V(l+mL)=0\quad a.s., \forall l\in\{0,1,\ldots, L-1\},
    \end{equation*}
    which implies
    \begin{equation*}
         \lim_{k\rightarrow\infty}V(k)=0\quad a.s.
    \end{equation*}
    Then by Theorem 4.4.6 in~\cite{durrett2019probability}, we have
    $\sum_{t=0}^k\alpha(t)\sum_{i=1}^N w_i(t)$ converges almost surely as $k\rightarrow\infty$. Together with \eqref{ieq:3} and \eqref{ieq:4}, we have $x_i(k),\forall i\in\mathcal{V}$ converge almost surely as $k\rightarrow\infty$.
    
\end{document}